\documentclass[10pt,journal,twocolumn]{IEEEtran}    
\usepackage{graphicx}
\usepackage{booktabs}
\usepackage{amsfonts}
\usepackage{amsmath} 
\usepackage{amssymb}
\usepackage{amsthm} 
\usepackage{balance}
\usepackage{psfrag}
\usepackage{algpseudocode}
\usepackage{algorithm}
\usepackage{tikz}
\usepackage{tikz-qtree}
\usetikzlibrary{trees}
\usetikzlibrary{fit,calc,arrows,positioning}
\usepackage{pgfplots}
\pgfplotsset{compat=newest}
\pgfplotsset{plot coordinates/math parser=false}
\newlength\figureheight
\newlength\figurewidth
\usepackage{graphicx}
\usepackage{subcaption}
\usepackage{caption}
\usepackage{mathtools}
\usetikzlibrary{shapes,arrows}
\makeatletter
\def\BState{\State\hskip-\ALG@thistlm}
\makeatother
\usepackage{tikz}
\usetikzlibrary{fit,calc,arrows,positioning}
\usepackage{pgfplots}
\pgfplotsset{compat=newest}
\pgfplotsset{plot coordinates/math parser=false}
\newcommand*{\QEDB}{\hfill\ensuremath{\square}}%
\newcommand{\range}{\operatorname{rge}}

\newcommand{\ess}{\operatorname{ess.}}
\newcommand{\thetatilde}{\tilde{\theta}}

\newcommand{\R}{\mathbb{R}}

\newcommand{\dom}{\operatorname{dom}}

\newcommand{\Id}{{I}}
\newcommand*{\tr}{^{\mkern-1.5mu\mathsf{T}}}
\newcommand*{\mtr}{^{\mkern-1.5mu\mathsf{-T}}}
\newcommand{\minimize}{\operatorname{minimize}}

\newcommand{\Co}{\operatorname{Co}}
\newcommand{\He}{\operatorname{He}}
\newcommand{\Diag}{\operatorname{diag}}

\newcommand{\Spn}{\mathcal{S}^n_{+}}
\newcommand{\Spnz}{\mathcal{S}^{n_z}_{+}}
\newcommand{\Spny}{\mathcal{S}^{n_y}_{+}}
\newcommand{\0}{{0}}
\newcommand{\eg}{{ e.g.}}
\newcommand{\ie}{{i.e.}}

\newcommand{\ep}{{\varepsilon}}

\newtheorem{theorem}{Theorem}

\newtheorem{corollary}{Corollary}

\newtheorem{definition}{Definition}
\newtheorem{example}{Example}

\newtheorem{lemma}{Lemma}
\newtheorem{fact}{Fact}
\newtheorem{assumption}{Assumption}

\newtheorem{problem}{Problem}
\newtheorem{proposition}{Proposition}
\newtheorem{remark}{Remark}

\title{\LARGE \bf $\mathcal{L}_2$ State Estimation with Guaranteed Convergence Speed in the Presence of Sporadic Measurements (Extended Version)}
\author{Francesco~Ferrante,
        Fr\'ed\'eric~Gouaisbaut, Ricardo~G.~Sanfelice 
        and~Sophie~Tarbouriech
        \thanks{Francesco Ferrante is with Univ. Grenoble Alpes, CNRS,  GIPSA-lab, F-38000 Grenoble, France. Email: francesco.ferrante@gipsa-lab.fr.} 
\thanks{Ricardo G. Sanfelice is with Department of Electrical and Computer Engineering, University of California,
Santa Cruz, CA 95064.
Email: ricardo@ucsc.edu}
\thanks{Fr\'ed\'eric Gouaisbaut, and Sophie Tarbouriech  are with LAAS-CNRS, Universit\'e de Toulouse, UPS, CNRS, Toulouse, France. Email:\{fgouaisb, tarbour\}@laas.fr}
\thanks{Research by R. G. Sanfelice has been partially supported by the National Science Foundation under CAREER Grant no. ECS-1450484, Grant no. ECS-1710621, and Grant no. CNS-1544396, by the Air Force Office of Scientific Research under Grant no. FA9550-16-1-0015, and by the Air Force Research Laboratory under Grant no. FA9453-16-1-0053.
}
}
\begin{document}
\maketitle
\begin{abstract}
This paper deals with the problem of estimating the state of a linear time-invariant system in the presence of sporadically available measurements and external perturbations.  An observer with a continuous intersample injection term is proposed.  Such an intersample injection is provided by a linear dynamical system, whose state is reset to the measured output estimation error at each sampling time.
The resulting system is augmented with a timer triggering the arrival of a new measurement and analyzed in a hybrid system framework. The design of the observer is performed to achieve global exponential stability with a given decay rate to a set wherein the estimation error is equal to zero. Robustness with respect to external perturbations and $\mathcal{L}_2$-external stability from the plant perturbation to a given performance output are considered. Moreover, computationally efficient algorithms based on the solution to linear matrix inequalities are proposed to design the observer. Finally, the effectiveness of the proposed methodology is shown in three examples.
\end{abstract}
\section{Introduction}
\subsection{Background}
In most real-world control engineering applications,  measurements of the output of a continuous-time plant are only available to the algorithms at isolated times.
Due to the use of digital systems in the implementation of the controllers, such a constraint is almost unavoidable and has lead researchers to propose algorithms that can cope with information not being available continuously.
In what pertains to state estimation, such a practical need has brought to life a new research area aimed at developing observer schemes accounting for the discrete nature of the available measurements. When the information is available at periodic time instances, 
there are numerous  design approaches in the literature 
that consist of designing a discrete-time observer for the discretized version of the process; see, e.g., \cite{arcak2004framework, nevsic1999formulas}, just to cite a few. 
Unfortunately, such an approach is limiting for several reasons. One reason stems from the fact that to precisely characterize the intersample behavior, one needs the exact discretized model of the plant, which may actually be impossible
to obtain analytically in the case of nonlinear systems; see \cite{nevsic1999formulas}.
 Furthermore, with such an approach no mismatch between the actual sampling time and the one used to discretize the plant  is allowed in the analysis or in the discrete-time model used to solve the estimation problem. Very importantly, many modern applications, such as network control
systems \cite{hespanha2007survey}, the output of the plant is often accessible only sporadically, making the fundamental assumption of measuring it periodically unrealistic.

To overcome the issues mentioned above, several state estimation strategies that accommodate information being available sporadically, at isolated times, have been proposed in the literature. Such strategies essentially belong to two main families. The first family pertains to observers whose state is entirely reset, according to a suitable law, whenever a new measurement is available, and that run open-loop in between such events -- these are typically called {\em continuous-discrete observers}. The design of such observers is pursued, \eg, in \cite{Ferrante2016state,mazenc2015design}. 
In particular, in \cite{Ferrante2016state} the authors propose a hybrid systems approach to model and design, via \emph{Linear Matrix Inequalities} (LMIs), a continuous-discrete observer ensuring exponential convergence of the estimation error and input-to-state stability with respect to measurement noise. 
In \cite{mazenc2015design}, a new design for continuous-discrete observers based on cooperative systems is proposed for the class of Lipschitz nonlinear systems.

The second family of strategies pertains to continuous-time observers whose output injection error between consecutive measurement events is estimated via a continuous-time update of the latest output measurement. This approach is pursued in \cite{farza2014continuous, karafyllis2009continuous, postoyan2012framework, postoyan2014tracking,raff2008observer}. Specifically, the results in \cite{karafyllis2009continuous,farza2014continuous} show that if a system admits a continuous-time observer and the observer has suitable robustness properties, then, one can build an observer guaranteeing asymptotic state reconstruction in the presence of intermittent measurements, provided that the time in between measurements is small enough.
Later, the general approach in \cite{karafyllis2009continuous} has been also extended by \cite{postoyan2012framework} to the more general context on networked systems, in which communication protocols are considered.
A different approach is pursued in \cite{raff2008observer}. In particular, in this work, the authors, building on the literature of sampled-data systems, propose sufficient conditions in the form of LMIs to design a sampled-and-hold observer to estimate the state of a Lipschitz nonlinear system in the presence of sporadic measurements.  
\subsection{Contribution}
In this paper,  we consider the problem of exponentially estimating the state of  continuous-time Lipschitz nonlinear systems subject to external disturbances and in the presence of sporadic measurements, \ie, we assume the plant output to be sampled with a bounded nonuniform sampling period, possibly very large.  
To address this problem, we propose an observer with a continuous intersample injection and state resets.  Such an intersample injection is provided by a linear time-invariant system, whose state is reset to the measured output estimation error at each sampling time.

Our contributions in the solution to this problem are as follows. Building on a hybrid system model of the proposed observer and of its interconnection with the plant, we propose results for the simultaneous design (co-design) of the observer and the intersample injection dynamics for the considered class of nonlinear systems.  
The approach we pursue relies on Lyapunov theory for hybrid systems in the framework in \cite{goebel2012hybrid}; similar Lyapunov-based  analyses for observers are also available in \cite[Section VIII]{postoyan2014tracking}, \cite{wang2017observer,ahmed2012high}. The use of the hybrid systems framework \cite{goebel2012hybrid} can be seen as an alternative approach to the impulsive approach pursued, \eg, in \cite{farza2014continuous}.
The design we propose ensures exponential convergence of the estimation error with guaranteed convergence speed and robustness with respect to measurement noise and plant perturbations. More precisely, the decay rate of the estimation error can be specified as a design requirement cf. \cite{fichera2012convex}. In addition, for a given performance output, we propose conditions to guarantee a particular $\mathcal{L}_2$-gain between the disturbances entering the plant and the desired performance output.
The conditions in these results are turned into matrix inequalities, which are used to derive efficient design procedures of the proposed observer.

The methodology we propose gives rise to novel observer designs and allows one to recover as special cases the schemes presented in \cite{karafyllis2009continuous, raff2008observer}.

The remainder of the paper is organized as follows. 
Section~\ref{sec:ProblemStatement} presents the system under consideration, the state estimation problem we solve, the outline of the proposed observer, and the hybrid modeling of the proposed observer. Section~\ref{sec:Design} is dedicated to the design of the proposed observer and to some optimization aspects.
Finally, in an example, Section~\ref{sec:Examples} shows the effectiveness of the results presented.  A preliminary version of the results here appeared in the conference paper \cite{ferrante2015hybrid}.
\vspace{0.1cm}
  
\noindent{\bf Notation}:
The set $\mathbb{N}$ is the set of positive integers including zero, the set $\mathbb{N}_{>0}$ is the set of strictly positive integers, $\mathbb{R}_{\geq 0}$ represents the set of nonnegative real scalars, $\R^{n\times m}$ represents the set of the $n\times m$ real matrices, and $\Spn$ is the set of $n\times n$ symmetric positive definite matrices.
The identity matrix is denoted by $\Id$, whereas the null matrix is denoted by $\0$.
For a matrix $A\in\mathbb{R}^{n\times m}$, $A\tr$ denotes the transpose of $A$, $A\mtr=(A\tr)^{-1}$, and $\He (A)= A+A^{\tr}$. For a symmetric matrix $A$, $A>0$ and $A\geq 0$ ($A<0$ and  $A\leq 0$) mean that $A$ ($-A$) is, respectively, positive definite and positive semidefinite. In partitioned symmetric matrices, the symbol $\bullet$ stands for symmetric blocks. Given matrices $A$ and $B$, the matrix $A\oplus B$ is the block-diagonal matrix having $A$ and $B$ as diagonal blocks.
For a vector $x\in\mathbb{R}^n$, $\vert x \vert$ denotes the Euclidean norm. Given two vectors $x,y$, we denote $(x,y)=[x'\,\,y']'$.
Given a vector $x\in\mathbb{R}^{n}$  and a closed set $\mathcal{A}$, the distance of $x$ to $\mathcal{A}$ is defined as 
$\vert x \vert_{\mathcal{A}}=\inf_{y\in {\mathcal{A}}} \vert x-y \vert$. For any function $z:\mathbb{R}\rightarrow\mathbb{R}^n$, we denote $z(t^+)\coloneqq \lim_{s\rightarrow t^+} z(s)$ when it exists. 
\subsection{Preliminaries on Hybrid Systems}
\label{sec:PreliminariesHybrid}
We consider hybrid systems with state $x\in\R^{n_x}$, input $u=(w, \eta)\in\R^{n_u}$, and output $y\in\R^{n_y}$ of the form 
$$
\mathcal{H}\left\lbrace
\begin{array}{ccll}
\dot{x}&=&f(x, w)&\quad x\in C\\
x^+&\in&G(x, \eta)&\quad x\in D\\
y&=&h(x)&
\end{array}\right.
$$
In particular we denote, $f\colon\R^{n_x}\rightarrow\R^{n_x}$ as the \emph{flow map}, $C\subset\R^{n_x}$ as the \emph{flow set}, $G\colon\R^{n_x}\rightrightarrows\R^{n_x}$ as the \emph{jump map}, and $D\subset\R^{n_x}$ as the \emph{jump set}.

A set $E\subset\R_{\geq 0}\times \mathbb{N}$ is a \emph{hybrid time domain} if it is the union of a finite or infinite sequence of intervals $[t_j,t_{j+1}]\times\{j\}$, with the last interval (if existent) of the form $[t_j,T)$ with $T$ finite or $T=\infty$. 
Given a hybrid time domain $E$, we denote $\sup_j E=\sup \{j\in\mathbb{N}\colon \exists t\in\R_{\geq 0}\,\mbox{s.t.}\,(t,j)\in E\}$. 
A hybrid signal $\phi$ is a function defined over a hybrid time domain.
Given a hybrid signal $w$, then $\dom_t w\coloneqq \{t\in\R_{\geq 0}\colon \exists j\in\mathbb{N}\,\,\mbox{s.t.}\,\,(t,j)\in\dom w\}$.
A hybrid signal $u\colon \dom u\rightarrow \R^{n_u}$ is called a hybrid input if $u(\cdot, j)$ is measurable and locally essentially bounded for each $j$. In particular, we denote $\mathcal{U}^{n_u}$ the class of hybrid inputs with values in $\R^{n_u}$. A hybrid signal $\phi\colon \dom\phi\rightarrow\R^{n_x}$ is a hybrid arc if $\phi(\cdot, j)$ is locally absolutely continuous for each $j$. In particular, we denote $\mathcal{X}^{n_x}$ the class of hybrid arcs with values in $\R^{n_x}$. Given a hybrid signal $u$, $j(t)=\min\{j\in\mathbb{N}\colon\,\, (t,j)\in\dom u\}$.
A hybrid arc $\phi\in\mathcal{X}^{n_x}$ and a hybrid input $u\in\mathcal{U}^{n_u}$ define a solution pair $(\phi,u)$ to $\mathcal{H}$ if  $\dom\phi=\dom u$ and $(\phi,u)$ satisfies the  dynamics of $\mathcal{H}$. A solution pair $(\phi,u)$ to $\mathcal{H}$ is maximal if it cannot be extended and is complete if $\dom\phi$ is unbounded; see~\cite{cai2009characterizations} for more details. With a slight abuse of terminology, given $\tilde{u}\in\mathcal{L}_\infty^{loc}(\R_{\geq 0},\R^{n_u})$, in the sequel we say that $\tilde{u}$ leads to a solution $\phi$ to $\mathcal{H}$ if $(\phi, u)$, with $u(t,j)=\tilde{u}(t)$ for each $(t,j)\in\dom\phi$, is a solution pair to $\mathcal{H}$. 
\section{Problem Statement and Outline of Proposed Observer}
\label{sec:ProblemStatement}
\subsection{System Description}
We consider continuous-time nonlinear time-invariant systems with disturbances of the form
\begin{equation}
\begin{aligned}
\label{eq:P2:Chap3:Plant}
&\dot{z}=Az+B\psi(Sz)+Nw,\quad y=Cz+\eta
\end{aligned}
\end{equation}
where $z\in\mathbb{R}^{n_z}, y\in\mathbb{R}^{n_y}$, $w\in\mathbb{R}^{n_w}$, and $\eta\in\R^{n_y}$ are, respectively, the state,  the measured output of the system, a nonmeasurable exogenous input, and the measurement noise affecting the output $y$, while $\psi\colon\R^{n_q}\rightarrow\R^{n_s}$ is a Lipschitz function with Lipschitz constant 
$\ell>0$, \ie, for all $v_1,v_2\in\R^{n_q}$
\begin{equation}
\label{eq:LipsBound}
\vert\psi(v_1)-\psi(v_2)\vert\leq \ell\vert v_1-v_2\vert
\end{equation}
The matrices $A,C, B, S$, and $N$ are constant and of appropriate dimensions. 
The output $y$ is available only at some time instances $t_k$, $k\in\mathbb{N}_{>0}$, not known {\itshape a priori}.
We assume that the sequence $\{t_k\}^{\infty}_{k=1}\!$ is strictly increasing and unbounded, and that (uniformly over such sequences) there exist two positive real scalars $T_1\leq T_2$  such that  
\begin{equation}
\begin{array}{lr}
\label{eq:P2:Chap3:timebound}

0\leq t_{1}\leq T_2,\qquad T_1\leq t_{k+1}-t_{k}\leq T_2\quad\forall k\in\mathbb{N}_{>0}
\end{array}
\end{equation}
The lower bound in condition \eqref{eq:P2:Chap3:timebound} prevents the existence of accumulation points in the sequence $\{t_k\}^{\infty}_{k=1}$, and, hence, avoids the existence of Zeno behaviors, which are typically undesired in practice.
In fact, $T_1$ defines a strictly positive minimum time in between consecutive measurements. Furthermore, $T_2$ defines the \emph{Maximum Allowable Transfer Time} (MATI)  \cite{postoyan2012framework}.
 
Given a performance output $y_p\coloneqq C_p(z-\hat{z})$, where $\hat{z}$ is the estimate of $z$ to be generated, the problem to solve is as follows: 
 \begin{problem}
 		\label{prob:Problem1}
 	Design an observer providing an estimate $\hat{z}$ of $z$, such that the following three properties are fulfilled:
\begin{itemize}
 \item[($P1$)] The set of points where the plant state $z$ and its estimate $\hat{z}$ coincide (and any other state variables\footnote{The observer may have extra state variables that are used for estimation. In our setting, the sporadic nature of the available measurements of $y$ will be captured by a timer with resets.} are bounded) is globally exponentially stable with a prescribed convergence rate for the plant \eqref{eq:P2:Chap3:Plant} interconnected with the observer whenever the input $w$ and $\eta$ are identically zero; 
 		\item[($P2$)] The estimation error is bounded when the disturbances $w$ and $\eta$ are bounded; 
 		\item[($P3$)]$\mathcal{L}_2$-external stability from the input $w$ to the performance output $y_p$ is ensured with a prescribed $\mathcal{L}_2$-gain when $\eta\equiv 0$.
 	\end{itemize}
 \end{problem}
 \subsection{Outline of the Proposed Solution}
Since measurements of the output $y$ are available in an impulsive fashion, assuming that the arrival of a new measurement can be instantaneously detected,   inspired by \cite{karafyllis2009continuous,postoyan2012framework,raff2008observer} to solve Problem~\ref{prob:Problem1}, we propose the following observer with jumps
\begin{equation}
\label{eq:P2:Chap3:ObsSampleHold}
\begin{array}{ll}
\left.
\begin{array}{rcl}
\dot{\hat{z}}(t)& =& A\hat{z}(t)+B\psi(S\hat{z}(t))+L\theta(t)\\
\dot{\theta}(t)& =&H\theta(t)
\end{array}\hspace{-0.1cm}\right\} &\hspace{-0.3cm} \forall t\neq t_k,k\in\mathbb{N}_{>0}\\\\
\left.
\begin{array}{rcl}
\hat{z}(t^{+}) & = &\hat{z}(t)\\
\theta(t^{+})& = &y(t)-C\hat{z}(t)
\end{array} \right\}&
\hspace{-0.3cm} \forall t=t_k,k\in\mathbb{N}_{>0}
\quad\begin{array}{rcl}
\end{array} 
\end{array}
\end{equation}
where $L$ and $H$ are real matrices of appropriate dimensions to be designed and $\hat{z}$ represents the estimate of $z$ provided by the observer.
The operating principle of the observer in \eqref{eq:P2:Chap3:ObsSampleHold} is as follows.  The arrival of a new measurement triggers an instantaneous jump in the observer state. Specifically, at each jump, the measured output estimation error, {\itshape i.e.}, $e_y\coloneqq y-C\hat{z}$, is instantaneously stored in $\theta$.  Then, in between consecutive measurements, $\theta$ is continuously updated according to continuous-time dynamics, and its value is continuously used as an intersample correction to feed a continuous-time observer. At this stage, we introduce the following change of variables
$\varepsilon\coloneqq z-\hat{z}, \tilde{\theta}\coloneqq C(z-\hat{z})-\theta$
which defines, respectively, the estimation error and the difference between the output estimation error and $\theta$. Moreover, by defining as a performance output $y_p=C_p \ep$, where $C_p\in\R^{n_{y_p}\times n_z}$, we consider the following dynamical system with jumps:
\begin{equation}
\begin{array}{ll}
\left\lbrace
\begin{array}{ll}
\dot{z}(t)&=\hspace{-0.0cm}Az(t)+B\psi(S z(t))+N w(t)\\
\begin{pmatrix}
\dot{\ep}(t)\\
\dot{\tilde{\theta}}(t)
\end{pmatrix}\hspace{-0.0cm}&=\hspace{-0.0cm}\mathcal{F}\begin{pmatrix}
\ep(t)\\
\tilde{\theta}(t)
\end{pmatrix}+\mathcal{Q}\zeta(z(t), \ep(t))+\mathcal{T}w(t)
\end{array}\right. &\hspace{-0.45cm}\forall t\neq t_k\\\\
\left\lbrace
\begin{array}{ll}
z(t^+)&=z(t)\\
\begin{pmatrix}
\ep(t^+)\\
\tilde{\theta}(t^+)
\end{pmatrix}&=\mathcal{G}\begin{pmatrix}
\ep(t)\\
\tilde{\theta}(t)
\end{pmatrix}+\mathcal{N}\eta(t)
\end{array}\hspace{0.3cm} \right.&
\hspace{-0.45cm}\forall t=t_k\\\\
\begin{array}{rcl}
y_p(t)&=C_p\ep(t)
\end{array} 
\end{array}
\label{eq:P2:Chap3:ObsGenErr}
\end{equation}
where for each $v_1, v_2\in\R^{n_z}$, $\zeta(v_1,v_2)\coloneqq \psi(Sv_1)-\psi(S(v_1-v_2))$ and
\begin{equation}
\label{eq:P2:Chap3:FG}
\begin{array}{ll}
\mathcal{F}\coloneqq\begin{pmatrix}
A-LC&L\\
CA-CLC-HC&CL+H
\end{pmatrix},\mathcal{T}\coloneqq\begin{pmatrix}
N\\
CN
\end{pmatrix}\\
\mathcal{Q}\coloneqq\begin{pmatrix}
B\\
C B
\end{pmatrix},\quad\mathcal{G}\coloneqq\begin{pmatrix}
\Id&\0\\
\0&\0
\end{pmatrix},\quad\,\,\,\,\qquad\mathcal{N}\coloneqq\begin{pmatrix}
\0\\-\Id
\end{pmatrix}
 \end{array}
\end{equation}
Our approach consists of recasting \eqref{eq:P2:Chap3:ObsGenErr} and the events at instants $t_k$ satisfying \eqref{eq:P2:Chap3:timebound} as a hybrid system with nonunique solutions and then apply hybrid systems theory to guarantee that \eqref{eq:P2:Chap3:ObsGenErr} solves Problem~\ref{prob:Problem1}.
\begin{remark}
As a difference to \cite{farza2014continuous, karafyllis2009continuous,postoyan2012framework}, the results presented in the next two sections are based on the Lyapunov results for hybrid systems presented in \cite{goebel2012hybrid} and, rather than emulation, consist of direct design methods of the proposed hybrid observer. Our design methods not only allow for completely designable intersample injection terms in the observer, but also allow for designs that cover the special cases of the schemes presented in \cite{karafyllis2009continuous,raff2008observer}. Furthermore, as a difference to \cite{postoyan2012framework}, where an emulation-based approach is considered, our results provide constructive conditions for the design of the observer gains so as to enforce the desired convergence properties for a desired value of $T_2$.	
\end{remark}
\section{Construction of the Observer and First Results}
\label{sec:Design}
\subsection{Hybrid Modeling}
The fact that the observer experiences jumps when a new measurement is available and evolves according to a differential equation in between updates suggests that the updating process of the error dynamics can be described via a hybrid system.  Due to this, we represent the whole system composed by the plant \eqref{eq:P2:Chap3:Plant}, the observer \eqref{eq:P2:Chap3:ObsSampleHold}, and the logic triggering jumps as a hybrid system.
The proposed hybrid systems approach also models the hidden time-driven mechanism triggering the jumps of the observer.

To this end, in this work, and as in \cite{FerranteIFAC2014}, we augment the state of the system  with an auxiliary timer variable $\tau$ that keeps track of the duration of flows and triggers a jump whenever a certain condition is verified. This additional state allows to describe the time-driven triggering mechanism as a state-driven triggering mechanism, which leads to a model that can be efficiently represented by relying on the framework for hybrid systems proposed in \cite{goebel2012hybrid}.
More precisely, we make $\tau$ decrease as ordinary time $t$ increases and, whenever $\tau=0$, reset it to any point in $[T_1,T_2]$, so as to enforce \eqref{eq:P2:Chap3:timebound}. After each jump, we allow the system to flow again.  
The whole system composed by the states $z$, $\ep$ and $\tilde{\theta}$, and the timer variable $\tau$ can be represented by the following hybrid system, which we denote $\mathcal{H}_e$, with state
$$x=(z, \ep,\tilde{\theta},\tau)\in\R^{n_x}$$
with  $n_x\coloneqq 2n_z+n_y+1$, input $u=(w,\eta)\in\R^{n_u}$, $n_u\coloneqq n_w+n_y$, and output $y_p$:
\begin{subequations}
\begin{equation}
\label{eq:P2:Chap3:ObsHybrid}
\begin{array}{l}
\left\{\begin{array}{ll}
\begin{array}{rcl}
\dot{x}&=&f(x,w)
\end{array}&\hspace{-0.33cm}x\in \mathcal{C}, w\in\R^{n_w}\\
\begin{array}{rcl}
x^+&\in&G(x,\eta)\\
\end{array}&\hspace{-0.34cm}x\in \mathcal{D}, \eta\in\R^{n_y}
\end{array}
\right.\\
\qquad\begin{array}{rcl}
y_p&=& C_p\ep
\end{array}
\end{array}
\end{equation}
where
\begin{equation}
\label{eq:P2:Chap3:FlowMap}
f(x,w)=\left(\begin{smallmatrix}
Az+ B\psi(Sz)+Nw\\
\mathcal{F}\begin{pmatrix}
\ep\\
\tilde{\theta}
\end{pmatrix}+\mathcal{Q}\zeta(z,\ep)+\mathcal{T}w\\
-1
\end{smallmatrix}\right)\qquad\forall x\in \mathcal{C}, w\in\R^{n_w}
\end{equation}
\begin{equation}
\label{eq:P2:Chap3:JumpMap}
G(x,\eta)=\left(\begin{smallmatrix}
z\\
	\mathcal{G}\begin{pmatrix}
		\ep\\
		\tilde{\theta}
	\end{pmatrix}+\mathcal{N}\eta\\
	[T_1,T_2]
\end{smallmatrix}\right)\qquad\forall x\in \mathcal{D}, \eta\in\R^{n_y}
\end{equation}
where the flow set $\mathcal{C}$ and the jump set $\mathcal{D}$ are defined as follows
\begin{equation}
\label{eq:P2:Chap3:Sets2}
\mathcal{C}=\mathbb{R}^{2n_z+n_y}\times[0,T_2],\quad\mathcal{D}=\mathbb{R}^{2n_z+n_y}\times\{0\}.
\end{equation}
\end{subequations}
The set-valued jump map allows to capture all possible sampling events occurring within $T_1$ or $T_2$ units of time from each other.
Specifically, the hybrid model in \eqref{eq:P2:Chap3:ObsHybrid} is able to characterize not only the behavior of the analyzed system for a given sequence $\{t_k\}_{k=1}^\infty$, but for any sequence satisfying \eqref{eq:P2:Chap3:timebound}. 

Concerning the existence of solutions to system \eqref{eq:P2:Chap3:ObsHybrid} with zero input, by relying on the notion of solution proposed in \cite{goebel2012hybrid}, it is straightforward to check that for every initial condition $\phi(0,0)\in \mathcal{C}\cup \mathcal{D}$ every maximal solution to \eqref{eq:P2:Chap3:ObsHybrid} is complete. Thus, completeness of the maximal solutions to \eqref{eq:P2:Chap3:ObsHybrid} is guaranteed for any choice of the gains $L$ and $H$, guaranteeing that $\mathcal{H}_e$ provides an accurate model of the error dynamics in \eqref{eq:P2:Chap3:ObsGenErr}.
In addition, one can characterize the domain of these solutions. 
Indeed for every initial condition $\phi(0,0)\in \mathcal{C}\cup \mathcal{D}$, the domain of every maximal solution $\phi$ to \eqref{eq:P2:Chap3:ObsHybrid} can be written as follows:
\begin{subequations}
\label{eq:Dom}
\begin{equation}
\dom \phi=\bigcup_{j\in\mathbb{N}}([t_j,t_{j+1}])\times\{j\}
\end{equation}
with $t_0=0$ and
\begin{equation}
\label{eq:Domb}
\begin{aligned}
0\leq t_{1}\leq T_2, \qquad
T_1  \leq t_{j+1}-t_j\,\leq T_2\qquad\forall j\in\mathbb{N}_{>0}
\end{aligned}
\end{equation}
\end{subequations}
where $\dom \phi$ is the domain of the solution $\phi$, which is a hybrid time domain; see \cite{goebel2012hybrid} for further details on hybrid time domains. 

Concerning solution pairs to \eqref{eq:P2:Chap3:ObsHybrid} with nonzero inputs, observe that given any solution pair $(\phi,u)$, the definition of the sets $\mathcal{C}$ and $\mathcal{D}$ ensure that 
$\dom\phi$ has the same structure illustrated in \eqref{eq:Dom}. Moreover, if $(\phi, u)$ is maximal then it is also complete\footnote{Completeness of maximal solution pairs can be shown by following similar arguments as in \cite[Proposition 6.10.]{goebel2012hybrid}. In particular, 
it is enough to observe that: $G(\mathcal{D})\subset \mathcal{C}$, no finite escape time is possible (due to $w$ measurable and locally essentially bounded and $x\mapsto f(x,w)$ Lipschitz uniformly in $w$), and solutions to $\dot{x}=f(x,w)$ from any initial condition in $\mathcal{ C}\setminus \mathcal{D}$ are nontrivial.}.  

To solve Problem~\ref{prob:Problem1} our approach is to design the matrices $L$ and $H$ in the proposed observer in \eqref{eq:P2:Chap3:ObsHybrid} such that without disturbances, \ie, $w\equiv 0, \eta\equiv 0$, the following set\footnote{By the definition of the system $\mathcal{H}_e$ and of the set $\mathcal{A}$, for every $x\in C \cup D \cup G(D)$, $\vert x \vert_\mathcal{A}=\vert (\ep,\tilde{\theta}) \vert$.}
\begin{equation}
\label{eq:P2:Chap3:A}
\mathcal{A}=\R^{n_z}\times\{0\}\times \{0\}\times[0,T_2]
\end{equation}
is exponentially stable and, when the disturbances are nonzero, the system $\mathcal{H}_e$ is input-to-state stable with respect to $\mathcal{A}$. These properties are captured by the notions defined below: 
\begin{definition}($\mathcal{L}_\infty$ norm)
Let $u$ be a hybrid signal and $T\in\R_{\geq 0}$. The $T$-truncated $\mathcal{L}_\infty$ norm of $u$ is given by
$$
\begin{aligned}
\Vert u_{[T]}\Vert_{\infty}&\coloneqq\max\left\{\underset{(s,k)\in\dom u\setminus\Gamma(u),s+k\leq T}{\ess\sup \vert u(s,k)\vert},\underset{(s,k)\in\Gamma(u),s+k\leq T}{\sup \vert u(s,k)\vert}\right\}
\end{aligned}$$
 where $\Gamma(u)$ denotes the set of all $(t,j) \in\dom u$ such that $(t,j + 1)\in\dom u$; see \cite{cai2009characterizations} for further details.
The $\mathcal{L}_\infty$ norm of $u$,  denoted by $\Vert u\Vert_{\infty}$ is given by $\lim_{T\rightarrow T^{\star}} \Vert u_{[T]}\Vert_{\infty}$, where $T^{\star}=\sup\{t+j\colon (t,j)\in \dom u\}$. When, in addition, $\Vert u\Vert_{\infty}$ is finite, we say that $u\in \mathcal{L}_\infty.$ 
\end{definition}
\begin{definition}[Pre-exponential input-to-state stability]
\label{defISS}
Let $\mathcal{A}\subset\R^{n_z+n_y+1}$ be closed. The system $\mathcal{H}_e$ is \emph{pre-exponentially input-to-state-stable} with respect to $\mathcal{A}$ if there exist $\kappa,\lambda>0$ and $\rho\in\mathcal{K}$ such that each solution pair~$(\phi,u)$ to $\mathcal{H}_e$ with $u\in\mathcal{L}_\infty$ satisfies 
\begin{equation}
\label{issBound}
\vert \phi(t,j)\vert_{\mathcal{A}}\leq \max\{\kappa e^{-\lambda (t+j)}\vert \phi(0,0)\vert_{\mathcal{A}}, \rho (\Vert u\Vert_{\infty})\}
\end{equation}
for each $(t,j)\in\dom\phi$. Whenever every maximal solution is complete, we say that $\mathcal{H}_e$ is exponentially input-to-state-stable (eISS) with respect to $\mathcal{A}$.
\end{definition}
\subsection{Sufficient conditions}
In this section we provide a first sufficient condition to solve Problem~\ref{prob:Problem1}. To this end, let us consider the following assumption, which is somehow driven by \cite[Example 27]{goebel2009hybrid} and whose role will be clarified later via Theorem~\ref{theorem:P2:Chap3:Main}.
\begin{assumption}
\label{Ass:L2}
Let $\lambda_t$ and $\gamma$ be given positive real numbers. There exist two continuously differentiable functions $V_1\colon\mathbb{R}^{n_z}\rightarrow \R$, $V_2\colon\R^{n_y+1}\rightarrow \R$, positive real numbers $\alpha_1,\alpha_2,\omega_1,\omega_2$ such that
\begin{itemize}
\item[(A1)]$\alpha_1\vert \ep\vert^2\leq V_1(\ep)\leq \alpha_2\vert \ep\vert^2\qquad \forall x\in \mathcal{C}$;
\item[(A2)]$\omega_1\vert \thetatilde\vert^2\leq V_2(\thetatilde, \tau)\leq \omega_2\vert \thetatilde\vert^2\quad \forall x\in \mathcal{C}$;
\item[(A3)] the function $x\mapsto V(x)\coloneqq V_1(\ep)+V_2(\thetatilde,\tau)$ satisfies for each $x\in \mathcal{C}, w\in\R^{n_w}$ 
\begin{equation}
\begin{split}
\label{eq:A3}
\langle\nabla V(x), 
\left(\begin{smallmatrix}
Az+ B\psi(Sz)+Nw\\
\mathcal{F}\begin{pmatrix}
\ep\\
\tilde{\theta}
\end{pmatrix}+\mathcal{Q}\zeta(z,\ep)+\mathcal{T}w\\
-1
\end{smallmatrix}\right)\rangle\leq -2\lambda_t V(x)\\
-\varepsilon\tr C_p\tr C_p \varepsilon+\gamma^2 w\tr w
\end{split}
\end{equation}\hfill $\triangle$
\end{itemize}
\end{assumption}
The following properties on the elements in the hybrid domain of solutions to $\mathcal{H}_e$ will be used to establish our sufficient conditions.
\begin{lemma}
	\label{lemma:tj}
	Let $\lambda_t>0$, $T_1>0$,  $\lambda\in\left(0,\frac{\lambda_t T_1}{1+T_1}\right]$, and $\omega\geq\lambda$.
	Then, each solution pair $(\phi,u)$ to $\mathcal{H}_e$ satisfies 
	\begin{equation}
	\label{eq:boundTJ}
	-\lambda_t t\leq \omega-\lambda(t+j)
	\end{equation}
	for every $(t,j)\in\dom\phi$.
\end{lemma}
\begin{proof}
	From \eqref{eq:boundTJ}, by rearranging the terms, one gets 
	\begin{equation}
	(-\lambda_t+\lambda)t+\lambda j-\omega\leq0.
	\label{eq:BoundTJ2}
	\end{equation}
	Now, pick any solution $\phi$ to hybrid system \eqref{eq:P2:Chap3:ObsHybrid}. From \eqref{eq:Domb}, it follows that for every $(t,j)\in\dom\phi$ 
	\begin{equation}
	\label{eq:BoundJ}
	j\leq \frac{t}{T_1}+1
	\end{equation}
	then, for every strictly positive scalar $\lambda$, from  the latter expression, and for every $(t,j)\in\dom\phi$, one gets
	\begin{equation}
	(-\lambda_t+\lambda)t+\lambda j-\omega\leq \left(-\lambda_t+\lambda+\frac{\lambda}{T_1}\right)t+\lambda-\omega.
	\end{equation}
	Thus, being $T_1$ strictly  positive,  by selecting 
	$$\lambda\in\left(0,\frac{\lambda_t T_1}{1+T_1}\right], \omega\geq \lambda$$
	yields \eqref{eq:BoundTJ2}, which concludes the proof.
\end{proof}
The following theorem shows that if there exist matrices $L\in\R^{n_z\times n_y}$ and 
$H\in\R^{n_y\times n_y}$ such that Assumption~\ref{Ass:L2} holds, then such matrices provide a solution to Problem~\ref{prob:Problem1}.
\begin{theorem}
\label{theorem:P2:Chap3:Main}
Let Assumption~\ref{Ass:L2} hold. Then: 
\begin{itemize}
\item[$(i)$] The hybrid system $\mathcal{H}_e$ is eISS with respect to $\mathcal{A}$;
\item[$(ii)$]  There exists $\alpha>0$ such that any solution pair $(\phi,u)$ to $\mathcal{H}_e$ with $\eta\equiv 0$ satisfies
$$
\begin{aligned}
\sqrt{\int_{\mathcal{I}} \vert y_p(s, j(s))\vert^2 ds}\leq &\alpha \vert \phi(0,0)\vert_{\mathcal{A}}+\\
&\gamma\sqrt{\int_{\mathcal{I}} \vert w(s, j(s))\vert^2ds}
\end{aligned}
$$
where $\mathcal{I}\coloneqq[0, \sup_t\dom\phi]\cap\dom_t \phi$;
\item [$(iii)$] The observer in \eqref{eq:P2:Chap3:ObsSampleHold} with $L$ and $H$ obtained from item (A3) in Assumption~\ref{Ass:L2} provides a solution to Problem~\ref{prob:Problem1}.
\end{itemize} 
\end{theorem}
\begin{proof}
Consider the following Lyapunov function candidate for the hybrid system \eqref{eq:P2:Chap3:ObsHybrid} defined for every $x\in\mathbb{R}^{2n_z+n_y}\times\mathbb{R}_{\geq 0}$:
 \begin{equation}
 \label{eq:P2:Chap3:V}
 V(x)=V_1(\varepsilon)+V_2(\tilde{\theta},\tau).
 \end{equation} 
 We prove $(i)$ first. 
To this end, notice that by setting $\rho_1=\min\{\alpha_1,\omega_1\}$ and $\rho_2=\max\{\alpha_2, \omega_2\}$, in view of the definition of the set $\mathcal{A}$ in \eqref{eq:P2:Chap3:A}, one gets
\begin{equation}
\label{eq:P2:Chap3:Sandwich}
\rho_1\vert x \vert_\mathcal{A}^2\leq V(x)\leq\rho_2\vert x \vert^2_\mathcal{A}\qquad\forall x\in \mathcal{C}\cup\mathcal{D}
\end{equation} 
Moreover, from Assumption~\ref{Ass:L2} item (A3) one has 
\begin{equation}
\label{eq:P2:Chap3:flow1}
\langle\nabla V(x),f(x,w)\rangle\leq -2\lambda_t V(x)+\gamma^2 w\tr w\qquad \forall x\in \mathcal{C}, w\in\R^{n_w}
\end{equation}
and for each 
$g=\left(z,\mathcal{G}\begin{pmatrix}\ep\\\thetatilde
\end{pmatrix}+\mathcal{N}\eta, v\right)\in G(x,\eta),x\in \mathcal{D}, \eta\in\R^{n_y}$ one has
\begin{equation}
\label{eq:P2:Chap3:jump2}
V(g)-V(x)=-V_2(\thetatilde,0)+V_2(-\eta,v)\leq \omega_2 \vert\eta\vert^2
\end{equation}
Let $(\phi,u)$ be a maximal solution pair to \eqref{eq:P2:Chap3:ObsHybrid} with $u=(w,\eta)\in \mathcal{L}_\infty$, and pick $(t, j)\in\dom\phi$. Furthermore, let $0=t_0\leq t_1\leq t_2\leq\dots\leq t_{j+1}=t$ be such that $\dom\phi\cap\left([0,t]\times\{0,1,\dots,j\}\right)=\cup_{i=0}^j \left([t_i, t_{i+1}]\times \{i\}\right)$.
Direct integration of $(t,j)\mapsto V(\phi(t,j))$ thanks to \eqref{eq:P2:Chap3:flow1} and \eqref{eq:P2:Chap3:jump2}, for each $(t,j)\in\dom\phi$, yields\footnote{Given a sequence $\{a_k\}$, we adopt the convention $\sum_{k=a}^b a_k=0$ if $a>b$.}
\begin{equation}
\label{eq:P2:Chap3:L2Bound}
\begin{split}
V(\phi(t,j))&\leq e^{-2\lambda_t t}V(\phi(0,0))\\
&+\gamma^2 e^{-2\lambda_t t}\int_{[0, t]\cap\dom_t\phi} e^{2\lambda_t s}\vert w(s,i)\vert^2 ds\\
&+\omega_2\sum_{i=1}^{j}e^{-2\lambda_t (t-t_i)}\vert\eta(t_i ,i-1)\vert^2
\end{split}
\end{equation} 
which in turns gives
\begin{equation}
\label{eq:P2:Chap3:BoundV}
\begin{split}
V(\phi(t,j))\leq &e^{-2\lambda_t t}V(\phi(0,0))+\frac{\gamma^2}{2\lambda_t} (1-e^{-2\lambda_t t})\Vert w\Vert^2_\infty\\
&+\omega_2\sum_{i=1}^{j}e^{-2\lambda_t (t-t_i)}\Vert \eta\Vert^2_\infty\\
&\quad\quad\quad\quad\quad\quad\quad\quad\quad\quad \forall (t,j)\in\dom\phi
\end{split}
\end{equation}
Now thanks to Lemma~\ref{lemma:KLj} in the Appendix, from \eqref{eq:P2:Chap3:BoundV} one gets for each $(t,j)\in\dom\phi$
$$
\begin{array}{ll}
V(\phi(t,j))&\leq e^{-2\lambda_t t}V(\phi(0,0))+\frac{\gamma^2}{2\lambda_t} \Vert w\Vert^2_\infty\\
&+\omega_2\frac{e^{4\lambda_t T_1}}{e^{2\lambda_t T_1}-1}\Vert \eta\Vert^2_\infty
\end{array}
$$
which, thanks to \eqref{eq:P2:Chap3:Sandwich}, implies that 
\begin{equation}
\begin{aligned}
\vert\phi(t,j)\vert^2_{\mathcal{A}}&\leq \frac{\rho_2}{\rho_1}e^{-2\lambda_t t}\vert\phi(0,0)\vert^2_{\mathcal{A}}+\frac{\gamma^2}{2\lambda_t\rho_1} \Vert w\Vert^2_\infty\\
&+\frac{e^{4\lambda_t T_1}}{(e^{2\lambda_t T_1}-1)\rho_1}\omega_2\Vert \eta\Vert^2_\infty\quad\forall (t,j)\in\dom\phi
\end{aligned}
\end{equation}
Hence, for each $(t,j)\in\dom\phi$ one has\footnote{The first inequality is established by using the fact that for each $a,b,$ and $c$ nonnegative real numbers, $\sqrt{a+b+c}\leq a^{\frac{1}{2}}+b^{\frac{1}{2}}+c^{\frac{1}{2}}$, while the second inequality follows from the fact that for any real numbers $a,b$, and $c$, $a\leq b+c$ implies $a\leq \max\{2b,2c\}$.} 
\begin{equation}
\begin{array}{ll}
\vert\phi(t,j)\vert_{\mathcal{A}}\!\!\leq\!\!&\sqrt{\frac{\rho_2}{\rho_1}}e^{-\lambda_t t}\vert\phi(0,0)\vert_{\mathcal{A}}+\frac{\gamma}{\sqrt{2\lambda_t\rho_1}} \Vert w\Vert_\infty\\
&+\sqrt{\omega_2\frac{e^{4\lambda_t T_1}}{e^{2\lambda_t T_1}-1}}\Vert \eta\Vert_\infty\\
&\leq \max\left\{2\sqrt{\frac{\rho_2}{\rho_1}}e^{-\lambda_t t}\vert\phi(0,0)\vert_{\mathcal{A}},2\max\{\frac{\gamma}{\sqrt{2\lambda_t\rho_1}},\right.\\
&\qquad\qquad\qquad\qquad\left.\sqrt{\omega_2\frac{e^{4\lambda_t T_1}}{e^{2\lambda_t T_1}-1}} \} \Vert u\Vert_\infty\right\}
\end{array}
\label{eq:t-eiss}
\end{equation}
Using Lemma~\ref{lemma:tj}, one gets that relation \eqref{issBound} holds with $\lambda\in\left(0,\frac{\lambda_t T_1}{1+T_1}\right]$,
 $\kappa=2\sqrt{\frac{\rho_2}{\rho_1}}e^\omega$, where  $\omega\geq\lambda$, and 
 $$
 s\mapsto\rho(s)\coloneqq 2\max\left\{\frac{\gamma}{\sqrt{2\lambda_t\rho_1}},\sqrt{\omega_2\frac{e^{4\lambda_t T_1}}{e^{2\lambda_t T_1}-1}} \right\} s
 $$
Hence, since every maximal solution to $\mathcal{H}_e$ is complete, $(i)$ is established.

To establish $(ii)$, we follow a similar approach as in \cite{nevsic2013finite}. 
Let $(\phi, u)$ be a maximal solution pair to $\mathcal{H}_e$ with $u=(w,0)$. Pick any $t>0$, then thanks to Assumption~\ref{Ass:L2} item (A3), since, as shown in \eqref{eq:P2:Chap3:jump2}, $V\circ\phi$ is nonincreasing at jumps, direct integration of $(t, j)\mapsto V(\phi(t, j))$ yields
\begin{equation}
\begin{split}
&V(\phi(t,j))-V(\phi(0,0))\leq\!\!-2\lambda_t \int_{\mathcal{I}(t)}  V(\phi(s,j(s)))ds\\
&-\int_{\mathcal{I}(t)}  \varepsilon(s,j(s))\tr C_p\tr C_p \varepsilon(s,j(s))ds+\gamma^2 \int_{\mathcal{I}(t)} \!\vert w(s,j(s))\vert^2 ds
\end{split}
\end{equation} 
where $\mathcal{I}(t)\coloneqq[0, t]\cap\dom_t \phi$,
which implies
\begin{equation}
\begin{split}
\int_{\mathcal{I}(t)} \varepsilon(s,j(s))\tr C_p\tr C_p \varepsilon(s,j(s))ds& \leq V(\phi(0,0))\\
&+\gamma^2\int_{\mathcal{I}(t)} \vert w(s,j(s))\vert^2 ds
\end{split}
\end{equation} 
Therefore, by taking the limit for $t$ approaching $\sup_t\dom\phi$, thanks to \eqref{eq:P2:Chap3:Sandwich}, one gets $(ii)$ with $\alpha=\rho_2$. 

To show that the proposed observer solves Problem~\ref{prob:Problem1} as claimed in item $(iii)$, we show that (P1), (P2), and (P3) are fulfilled. 
Item $(i)$ already implies (P1) and (P2), since $\lambda_t$ defines a lower bound on the decay rate with respect to the ordinary time $t$; see \eqref{eq:t-eiss}. To show that $(ii)$ implies (P3), notice that since $(ii)$ holds for any solution pair $(\phi, u)$ with $\eta\equiv 0$ and $w$ any hybrid signal,  it holds in particular when the hybrid signal $w$ is obtained from a continuous-time signal of the original plant \eqref{eq:P2:Chap3:Plant}. Passing from hybrid signals $w$ and $y_p$ to right continuous signals $\tilde{u},\tilde{y}_p$, respectively, (see \cite{Mayhew2013}), item $(ii)$ leads to 
\begin{equation}
\begin{array}{ll}
&\sqrt{\int_{\mathcal{I}} \vert y_p(s, j(s))\vert^2 ds}=\sqrt{\int_{\mathcal{I}}\vert\tilde{y}_p(s)\vert^2 ds}=\Vert \tilde{y}_p\Vert_2\\
&\leq \alpha \vert(\varepsilon_0, \thetatilde_0)\vert+\gamma\sqrt{\int_{\mathcal{I}}\vert\tilde{w}(s)\vert^2 ds}=\alpha \vert(\varepsilon_0, \thetatilde_0)\vert+\gamma\Vert \tilde{w}\Vert_2
\end{array}
\end{equation}
hence concluding the proof.
\end{proof}
\subsection{Construction of the functions $V_1$ and $V_2$ in Assumption~\ref{Ass:L2}}
With the aim of deriving constructive design strategies for the synthesis of the observer, we perform a particular choice for the functions $V_1, V_2$ in Assumption~\ref{Ass:L2}. 
Let $P_1\in\Spnz, P_2\in\Spny$, and $\delta$ be a positive real number. Inspired by \cite{FGZN_Auto2014},  we consider the following choice
\begin{equation}
\label{eq:V1V2}
V_1(\ep)=\ep\tr P_1\ep, \qquad  
V_2(\thetatilde, \tau)=e^{\delta\tau}\thetatilde\tr P_2\thetatilde 
\end{equation}
The structure selected above for the functions $V_1$ and $V_2$ essentially allows to exploit the (quasi)-quadratic nature of the resulting Lyapunov function candidate $x\mapsto V_1(\ep)+V_2(\thetatilde, \tau)$ to cast the solution to Problem~\ref{prob:Problem1} into the solution to certain matrix inequalities. 
\begin{theorem}
\label{Theorem1}
Let $\lambda_t$ and $\gamma$ be given positive real numbers. If there exist $P_1\in\Spnz,P_2\in\Spny$, positive real numbers $\delta, \chi$, and two matrices $L\in\mathbb{R}^{n_z\times n_y},H\in\mathbb{R}^{n_y\times n_y}$,  such that
\begin{equation}
\label{eq:M}
\mathcal{M}(0)\leq\0,\quad \mathcal{M}(T_2)\leq \0 
\end{equation}
where the function $[0, T_2]\ni\tau\mapsto\mathcal{M}(\tau)$ is defined in \eqref{eq:M1} (at the top of the next page), then Assumption~\ref{Ass:L2} holds. 
\end{theorem}
\begin{proof}
Let $V_1$ and $V_2$ be defined as in \eqref{eq:V1V2} and select $\alpha_1=\lambda_{\min}(P_1),\omega_1=\lambda_{\min}(P_2)$,
$\alpha_2=\lambda_{\max}(P_1)$, and $\omega_2=\lambda_{\max}(P_2)e^{\delta T_2}$. Then, it turns out that items (A1) and (A2) of Assumption~\ref{Ass:L2} are satisfied. Let $V(x)=V_1(\ep)+V_2(\thetatilde,\tau)$, then, by straightforward calculations and by the definition of the flow map in \eqref{eq:P2:Chap3:FlowMap}, it follows that for each $x\in \mathcal{C}$, $w\in\R^{n_w}$ one has 
\begin{equation}
\begin{split}
&\Omega(x,w)\coloneqq\langle\nabla V(x),f(x,w)\rangle+\ep\tr C_p\tr C_p\ep+2\lambda_t V(x)-\\
&\gamma^2 w\tr w=\ep\tr \He(P_1(A-LC))\ep+2\ep\tr P_1L\tilde{\theta}+\\
&e^{\delta\tau}\tilde{\theta}\tr\He(P_2(CL+H))\tilde{\theta}+\ep\tr C_p\tr C_p\ep+\\
&2e^{\delta\tau}\tilde{\theta}\tr P_2(CA-CLC-HC)\ep-\delta e^{\delta\tau}\tilde{\theta}\tr P_2\tilde{\theta}+\\
&2\ep\tr P_1Nw+2e^{\delta\tau}\thetatilde\tr P_2 C N w+2\lambda_t (\ep\tr P_1\ep+e^{\delta\tau}\thetatilde\tr P_2 \thetatilde)\\
&-\gamma^2 w\tr w+2\ep\tr P_1  B\zeta(z,\ep)+2e^{\delta\tau}\thetatilde\tr P_2 C B\zeta(z,\ep)
\end{split}
\end{equation}
Moreover, observe that thanks to \eqref{eq:LipsBound}, for any positive real number $\chi$ one has that 
$$
\Omega(x,w)\leq \underbrace{\Omega(x,w)-\chi (\zeta(z,\ep)\tr \zeta(z,\ep)-\ell^2 \ep\tr S\tr S\ep)}_{\Pi(x,w)}\,\,\forall x\in\mathcal{C}
$$
Therefore, by defining $\Psi(x,w)=(\varepsilon,\thetatilde, w,\zeta(z,\ep))$, for each $x\in \mathcal{C}, w\in\R^{n_w}$ one has
$\Omega(x,w)\leq \Pi(x,w)=\Psi(x,w)\tr\mathcal{M}(\tau)\Psi(z,w)$,
where the symmetric matrix $\mathcal{M}(\tau)$ is defined in \eqref{eq:M1}.
\begin{figure*}
{\small
\begin{equation}
\label{eq:M1}
\mathcal{M}(\tau)=\begin{pmatrix}
\He(P_1(A-LC))+2\lambda_t P_1+C_p\tr C_p+\chi\ell^2 S\tr S&\!P_1L+e^{\delta\tau}(CA-CLC-HC)\tr\!P_2&P_1N&P_1 B\\
\bullet&e^{\delta\tau}\!\left(\He(P_2(CL+H))+(2\lambda_t-\delta)P_2\right)&e^{\delta\tau}P_2CN&e^{\delta\tau} P_2 C B\\
\bullet&\bullet&-\gamma^2\Id_{n_w}&0\\
\bullet&\bullet&\bullet&-\chi\Id_{n_s}
\end{pmatrix}
\end{equation}}
\end{figure*}
To conclude this proof, notice that it is straightforward to show that there exists $\lambda\colon [0,\tau]\rightarrow[0,1]$ such that for each $\tau\in[0, T_2]$, $\mathcal{M}(\tau)=\lambda(\tau)\mathcal{M}(0)+(1-\lambda(\tau))\mathcal{M}(T_2)$; see Lemma~\ref{claim:RangeM}.
Therefore, it follows that the satisfaction of \eqref{eq:M} implies $\mathcal{M}(\tau)\leq\0$ for each $\tau\in [0, T_2]$. Hence, the result is established.
\end{proof}
\begin{remark}
Theorem~\ref{Theorem1} can be easily adapted to get a solution to Problem~\ref{prob:Problem1} for linear plants, \ie, when $\psi\equiv 0$. In particular, in such case a sufficient condition for the satisfaction of Assumption~\ref{Ass:L2} can be obtained by eliminating the forth row and the forth column from matrix $\mathcal{M}$ in Theorem~\ref{Theorem1} and by enforcing $\chi=0$.
\end{remark}
\begin{remark}
Notice that, for it to be feasible, condition \eqref{eq:M} requires the existence of $L\in\R^{n_z\times n_y}$ such that 
$\Vert \mathcal{T}\Vert_{\infty}\leq \gamma$, where $\mathbb{C}\ni s\mapsto \mathcal{T}(s)\coloneqq C_p (s\Id-(A-LC+\lambda_t\Id))^{-1}N$ and $\Vert\cdot \Vert_{\infty}$ stands for the $\mathcal{H}_\infty$ norm of its argument\footnote{To show this claim it suffices to observe that the satisfaction of \eqref{eq:M} implies
$\left(\begin{smallmatrix}
\He(P_1(A-LC))+2\lambda_t P_1+C_p\tr C_p&P_1N\\
\bullet&-\gamma^2\Id_{n_w}
\end{smallmatrix}\right)\leq \0$ which turns out to be equivalent to $\Vert \mathcal{T}\Vert_{\infty}\leq \gamma$; see \cite{Boyd}.}. Nevertheless, this condition is, in general, only necessary. 
\end{remark}
Although, for a given instance of Problem~\ref{prob:Problem1}, the search of feasible solutions to \eqref{eq:M} needs to be performed via numerical methods, it is worthwhile to provide minimum requirements to ensure,  at least for suitable values of $T_2,\lambda_t$ (small) and $\gamma$ (large), the feasibility of \eqref{eq:M}.
To this end, being  the satisfaction of \eqref{eq:M} equivalent to the satisfaction of item (A3) in Assumption~\ref{Ass:L2} (for the particular choice of the functions $V_1$ and $V_2$ in \eqref{eq:V1V2}),
one only needs to analyze under which conditions there exists a suitable selection of the real numbers $T_2,\lambda_t,\gamma$ that allows to fulfill (A3). This is illustrated in the result given next.
\begin{proposition}
\label{pro:Existence}
If there exist $L\in\mathbb{R}^{n_z\times n_y}$, $P_1\in\Spnz$, and $\lambda_t, \chi,\hat{\gamma}\in\R_{>0}$ such that
\begin{equation}
\label{eq:ExistenceLMI}
\left(\begin{smallmatrix}
\He(P_1(A-LC))+C_p\tr C_p+\chi\ell^2 S\tr S&P_1N&P_1 B\\
\bullet&-\hat{\gamma}^2\Id_{n_w}&\0\\
\bullet&\bullet&-\chi\Id_{n_s}
\end{smallmatrix}\right)<\0
\end{equation} 
Then, there exist four positive real numbers $T_2^\star,\gamma^\star,\delta,\lambda_t^\star$, and $P_2\in\Spny$ such that the function $x\mapsto V(x)\coloneqq \ep\tr P_1\ep+e^{\delta\tau}\thetatilde\tr P_2\thetatilde$ satisfies
$$\langle\nabla V(x), f(x,w)\rangle\leq -2\lambda_t^{\star} V(x)-\varepsilon\tr C_p\tr C_p \varepsilon+\gamma^{\star 2} w\tr w$$ 
for each $(x,w)\in\R^{2n_z+n_y}\times [0, T_2^\star]\times\R^{n_w}$.
\end{proposition}
\begin{proof}
From \eqref{eq:ExistenceLMI}, one has that there exist positive real numbers $\xi_1, \xi_2$ and a matrix $P_1\in\Spnz$ such that for each $(\ep,w)\in\R^{n_z+n_w}$
$$
\def\arraystretch{1.5}
\begin{array}{ll}
&\langle \nabla \underbrace{\ep\tr P_1\ep}_{V_1(\ep)}, \mathcal{F}_{11}\varepsilon+ B\zeta(z,\ep)+\mathcal{F}_{12}\thetatilde +Nw\rangle-\\
&\chi\zeta(z,\ep)\tr\!\zeta(z,\ep)+\chi\ell^2\ep\tr S\tr S\ep\leq -\xi_1 \ep\tr \ep+\hat{\gamma}^2 w\tr w-\\
&\ep\tr C_p\tr C_p\ep+2\ep\tr P_1\mathcal{F}_{12}\thetatilde-\xi_2 \zeta(z,\ep)\tr\zeta(z,\ep)
\end{array}
$$
which, by squares completion, gives
$$
\def\arraystretch{1.5}
\begin{array}{ll}
\langle \nabla \underbrace{\ep\tr P_1\ep}_{V_1(\ep)}, \mathcal{F}_{11}\varepsilon+ B\xi(z,\ep)+\mathcal{F}_{12}\thetatilde+ Nw\rangle-\chi\zeta(z,\ep)\tr\zeta(z,\ep)\\
+\chi\ell^2\ep\tr S\tr S\ep\leq-\frac{\xi_1 V_1(\ep)}{2\lambda_{\max}(P_1)}+\hat{\gamma}^2 w\tr w-\ep\tr C_p\tr C_p\ep\\
\hspace{2.1cm}+\frac{2\lambda_{\max}(P_1)\vert\mathcal{F}_{12}\tr P_1^2 \mathcal{F}_{12}\vert}{\xi_1}\thetatilde\tr\thetatilde-\xi_2 \zeta(z,\ep)\tr\zeta(z,\ep)
\end{array}
$$

Moreover, still by squares completion, for each $\beta_1,\beta_2,\beta_3\in\R_{>0}$ one has for every $(\ep,\thetatilde,w)\in\R^{n_z+n_y+n_w}$ and any $P_2\in\Spny$
$$
\begin{array}{ll}
&\langle \nabla \underbrace{\thetatilde\tr P_2\thetatilde}_{\widetilde{V_2}(\thetatilde)}, \mathcal{F}_{22}\thetatilde+\mathcal{F}_{21}\ep+C B\zeta(z,\ep)+CNw\rangle\leq \\
&\thetatilde\tr Q\thetatilde+\frac{1}{\beta_1}\ep\tr\ep\!+\!\frac{w\tr w}{\beta_2}\leq\frac{\lambda_{\max}(Q)}{\lambda_{\max}(P_2)}\widetilde{V_2}(\thetatilde)+
\frac{V_1(\varepsilon)}{\beta_1\lambda_{\min}(P_1)}\\
&\hspace{3.4cm}+\frac{1}{\beta_2}w\tr w+\frac{1}{\beta_3}\zeta(z,\ep)\tr \zeta(z,\ep)
\end{array}
$$
where $Q\coloneqq \He(P_2\mathcal{F}_{21})+P_2(\beta_1\mathcal{F}_{21}\mathcal{F}_{21}\tr+\beta_2 CN N\tr C\tr+\beta_3 C B  B\tr C\tr)P_2$. Therefore, for each $(x,w)\in \R^{2n_z+n_y}\times\R_{\geq 0}\times\R^{n_w}$ and any real positive number $\delta$ one has
$$
\def\arraystretch{1.5}
\begin{array}{ll}
&\langle \nabla V(x),f(x,w)\rangle-\chi(\zeta(z,\ep)\tr\zeta(z,\ep)-\ell^2\ep\tr S\tr S\ep)\leq\\ 
&V_1(\ep)\left(\frac{-\xi_1}{2\lambda_{\max}(P_1)}+\frac{e^{\delta \tau}}{\beta_1\lambda_{\min}(P_1)}\right)\\
&+V_2(\thetatilde)\left(-\delta+\frac{2e^{-\delta\tau}\lambda_{\max}(P_1)\vert\mathcal{F}_{12}\tr P_1^2 \mathcal{F}_{12}\vert}{\xi_1\lambda_{\min}(P_2)}+\frac{\lambda_{\max}(Q)}{\lambda_{\max}(P_2)}\right)\\
&+(\frac{1}{\beta_2} e^{\delta\tau}+\hat{\gamma}^2)w\tr w-\ep\tr C_p\tr C_p\ep+(\frac{1}{\beta_3}-\xi_2)\zeta(z,\ep)\tr\zeta(z,\ep)
\end{array}
$$
Pick $\beta_1$ large enough such that $\frac{\beta_1\lambda_{\min}(P_1)\xi_1}{2\lambda_{\max}(P_1)}\!>\!1$ and pick  $\frac{1}{\xi_2}=\beta_3$
then, by selecting 
$$
\begin{aligned}
&\delta>\frac{2\lambda_{\max}(P_1)\vert\mathcal{F}_{12}\tr P_1^2 \mathcal{F}_{12}\vert}{\xi_1\lambda_{\min}(P_2)}+\frac{\lambda_{\max}(Q)}{\lambda_{\max}(P_2)}\\
&\lambda_t^\star=\frac{1}{2}\min\left\{\left\vert\frac{-\xi_1}{2\lambda_{\max}(P_1)}+\frac{e^{\delta T_2^\star}}{\beta_1\lambda_{\min}(P_1)}\right\vert\right.,\\
&\qquad\left.\left\vert-\delta+\frac{2\lambda_{\max}(P_1)\vert\mathcal{F}_{12}\tr P_1^2 \mathcal{F}_{12}\vert}{\xi_1\lambda_{\min}(P_2)}+\frac{\lambda_{\max}(Q)}{\lambda_{\max}(P_2)}\right\vert\right\}\\
&T_2^\star<\frac{1}{\delta}\ln\left(\frac{\beta_1\lambda_{\min}(P_1)\xi_1}{2\lambda_{\max}(P_1)}\right), \gamma^{\star}=\sqrt{\frac{1}{\beta_2} e^{\delta T_2^\star}+\hat{\gamma}^2}
\end{aligned}
$$
which are all strictly positive real numbers, thanks to \eqref{eq:LipsBound},
one has for all $(x,w)\in\R^{2n_z+n_y}\times [0, T_2^\star]\times\R^{n_w}$,
$
\langle \nabla V(x),f(x,w)\rangle\!\!\leq -2\lambda_t^\star V(x)+\gamma^{\star 2}w\tr w-\ep\tr C_p\tr C_p\ep$,
concluding the proof.
\end{proof}
\begin{remark}
Essentially the above result shows that if a continuous-time Luemberger $\mathcal{L}_2$-gain observer, admitting a quadratic storage function, exists for \eqref{eq:P2:Chap3:Plant}, then, provided that $T_2$ is small enough, one can design via Theorem~\ref{Theorem1} an observer that solves Problem~\ref{prob:Problem1}. In the case of linear plants, the detectability of the pair $(A,C)$ is enough to guarantee the feasibility of the conditions stated in Theorem~\ref{Theorem1} adapted to the linear case.   
\end{remark}
In some applications, one may be interested in solving relaxed versions of Problem~\ref{prob:Problem1}, by for example, avoiding either to prescribe a certain decay rate or to consider a specific $\mathcal{L}_2$ gain. In such cases, the following results may be of interest. The proofs of such results consist of mere manipulations of the matrix inequalities \eqref{eq:M}, then, due to space limitations, we preferred to omit those proofs. 
\begin{corollary}
\label{Coro:WoGamma}
Let $\lambda_t$ be a given positive real number. If there exist $P_1\in\Spnz,P_2\in\Spny$, positive  real numbers $\delta, \chi$, and two matrices $L\in\mathbb{R}^{n_z\times n_y},H\in\mathbb{R}^{n_y\times n_y}$, such that 
\begin{equation}
\label{eq:M1M2Strict1}
\begin{aligned}
&\left(\begin{smallmatrix}
\Id&\0&\0&\0\\
\0&\Id&\0&\0\\
\0&\0&\0&\Id\\
\end{smallmatrix}\right)\mathcal{M}(0)\left(\begin{smallmatrix}
\Id&\0&\0\\
\0&\Id&\0\\
\0&\0&\0\\
\0&\0&\Id
\end{smallmatrix}\right)<\0, \left(\begin{smallmatrix}
\Id&\0&\0&\0\\
\0&\Id&\0&\0\\
\0&\0&\0&\Id\\
\end{smallmatrix}\right)\mathcal{M}(T_2)\left(\begin{smallmatrix}
\Id&\0&\0\\
\0&\Id&\0\\
\0&\0&\0\\
\0&\0&\Id
\end{smallmatrix}\right)<\0
\end{aligned}
\end{equation} 
then items $(i)$ and $(ii)$ in Theorem~\ref{theorem:P2:Chap3:Main} hold for some $\gamma>0$. \QEDB
\end{corollary} 
\begin{corollary}
\label{Coro:WoLambda}
Let $\gamma$ be a given positive real number. If there exist $P_1\in\Spnz,P_2\in\Spny$, positive real numbers $\delta, \chi$, and two matrices $L\in\mathbb{R}^{n_z\times n_y},H\in\mathbb{R}^{n_y\times n_y}$, such that 
\begin{equation}
\label{eq:M1M2Strict}
\begin{aligned}
&\mathcal{M}(0)\vert_{\lambda_t=0}<\0,\quad\mathcal{M}(T_2)\vert_{\lambda_t=0}<\0
\end{aligned}
 \end{equation} 
then $\mathcal{A}$ is eISS, for the hybrid system $\mathcal{H}_e$ with decay rate (with respect to flow time)
$\frac{\beta}{2\rho_2}$, where 
$$\beta=-\max_{\tau\in[0, T_2]}\lambda_{\max}(\mathcal{M}(\tau)\vert_{\lambda_t=0})$$
$\rho_2=\max\{\lambda_{\max}(P_1),\lambda_{\max}(P_2)e^{\delta T_2}\}$, and item $(ii)$ in Problem~\ref{prob:Problem1} hold. \QEDB
\end{corollary} 
\section{LMI-based observer design} 
\label{sec:LMI} 
In the previous section, sufficient conditions turning the solution to Problem~\ref{prob:Problem1} into the feasibility problem of certain matrix inequalities were provided. However, due to their form, such conditions are in general not computationally tractable to provide a viable solution to Problem~\ref{prob:Problem1}. Indeed, condition \eqref{eq:M} is nonlinear in the design variables $P_1,P_2,\delta, H$, and $L$; so further work is needed to derive a computationally tractable design procedure for the proposed observer.
Specifically, the nonlinearities present in \eqref{eq:M} are due to both the bilinear terms involving the matrices $P_1,P_2,L,H$, and the real number $\delta$, as well as the fact that $\delta$ also appears in a nonlinear fashion via the exponential function. From a numerical standpoint, the nonlinearities involving the real number $\delta$ are easily manageable in a numerical scheme by treating $\delta$ as a tuning parameter or selecting it via an iterative search. The main issue to tackle concerns with the other nonlinearities present in \eqref{eq:M}. To address these, in the sequel, we provide several sufficient conditions to solve Problem~\ref{prob:Problem1} via the solution to some linear matrix inequalities, whose solution can be performed in polynomial time through numerical solvers; see \cite{Boyd}.
\begin{proposition}
\label{prop:P2:Chap3:propPred}
Let $\lambda_t, \gamma$ be given positive real numbers. If there exist $P_1\in\Spnz,P_2\in\Spny$, positive real numbers $\delta,\chi$, matrices $J\in\mathbb{R}^{n_z\times n_y}$ and $Y\in\mathbb{R}^{n_y\times n_y}$ such that \eqref{eq:P2:Chap3:Coro1} (at the top of the next page) holds
\begin{figure*}
\begin{equation}
\label{eq:P2:Chap3:Coro1}
 \begin{aligned}
 &\left(\begin{smallmatrix}
\He(P_1A-JC)+2\lambda_t P_1+C_p\tr C_p+\ell^2\chi S
\tr S&J+A\tr C\tr P_2-C\tr Y&P_1N&P_1 B\\
\bullet&\He(Y)+(2\lambda_t-\delta)P_2&P_2 C N&P_2C B\\
\bullet&\bullet&-\gamma^2 \Id&\0\\
\bullet&\bullet&\bullet&-\chi\Id
\end{smallmatrix}\right)\leq\0\\
&\left(\begin{smallmatrix}
\He(P_1A-JC)+2\lambda_t P_1+C_p\tr C_p+\ell^2\chi S
\tr S&J+e^{\delta T_2}(A\tr C\tr P_2-C\tr Y)&P_1N&P_1 B\\
\bullet&(\He(Y)+(2\lambda_t-\delta) P_2)e^{\delta T_2}&e^{\delta T_2}P_2CN&e^{\delta T_2}P_2C B\\
\bullet&\bullet&-\gamma^2 \Id&\0\\
\bullet&\bullet&\bullet&-\chi\Id
\end{smallmatrix}\right)\leq\0
\end{aligned}
\end{equation}
\end{figure*}
then $L=P_1^{-1}J,H=P_2^{-1}Y\tr-CL$ is a solution to Problem~\ref{prob:Problem1}.
\end{proposition}
\begin{proof}
By setting $H=P_2^{-1}Y\tr-CL$ and $L=P^{-1}_1J$ in \eqref{eq:M} yields 
\eqref{eq:P2:Chap3:Coro1}, thus by the virtue of Theorem~\ref{Theorem1}, this concludes the proof.
\end{proof}
\begin{remark}
\label{rem:pred}
By selecting $Y=\0$, the above result leads to the predictor-based observer scheme proposed in \cite{0801.4824, karafyllis2009continuous}, though written in  different coordinates. 
Indeed, whenever $H=-CL$, by rewriting \eqref{eq:P2:Chap3:ObsSampleHold} via the following invertible change of variables $(\hat{z},w)=(\hat{z},\theta+C\hat{z})$, yields the same observer in \cite{0801.4824,karafyllis2009continuous}.
\end{remark}

The main idea behind the above result consists of selecting the design variable $H$ so as to cancel out the terms $CLC$ and the term involving the product of $P_2$ and $L$ (which would hardly lead to conditions linear in the decision variables).
\subsection{Slack Variables-Based Design}
Next, we present other design procedures, whose derivation is based on an equivalent condition to the ones in \eqref{eq:M} that is formulated introducing slack variables via the use of the projection lemma; see, \eg, \cite{gahinet1994linear,ebihara2015,Pipeleers:2009aa}.
Before stating the main result, let us consider the following fact.
\begin{fact}
\label{Claim}
The matrix $\mathcal{F}$ in \eqref{eq:P2:Chap3:FG} can be factorized as follows
\begin{equation}
\label{eq:Factorization}
\mathcal{F}=\underbrace{\begin{pmatrix}
\Id&\0\\
C&\Id
\end{pmatrix}}_{\mathcal{F}_l}\underbrace{\begin{pmatrix}
A-LC&L\\
-HC&H
\end{pmatrix}}_{\mathcal{F}_r}
\end{equation}
where $\mathcal{F}_l$ is nonsingular.
\end{fact}
Building on this fact, we have the following result.
\begin{theorem}
\label{Theorem:Projection}
Let $P_1\in\Spnz,P_2\in\Spny$, $H\in\R^{n_y \times n_y},L\in\R^{n_z\times n_y}$, and $\lambda_t,\gamma,\delta,\chi $ be strictly positive real numbers. The following statements are equivalent:
\begin{itemize}
\item[$(i)$] The matrix inequalities in \eqref{eq:M} are satisfied with strict inequalities;
\item[$(ii)$]There exist matrices 
$$
\begin{aligned}
&X_1,Y_1,X_3,Y_3\in\mathbb{R}^{n_z \times n_z}, X_2,X_4,Y_2,Y_4\in\mathbb{R}^{n_z \times n_y}\\
&X_5,Y_5,X_7,Y_7\in\mathbb{R}^{n_y \times n_z}, X_6,X_8,Y_6,Y_8\in\mathbb{R}^{n_y \times n_y}
\end{aligned}
$$ 
such that
\end{itemize}
\begin{equation}
\label{eq:ProjeCoro}
\begin{aligned}
&\left(\begin{smallmatrix}
\He(S_1(X))&S_2(X)+\mathcal{P}&S_3(X)&S_4(X)\\
\bullet&\mathcal{N}+\He(S_5(X))&S_6(X)&S_7(X)\\
\bullet&\bullet&-\gamma^2\Id&\0\\
\bullet&\bullet&\bullet&-\chi\Id
\end{smallmatrix}\right)<\0\\
&\left(\begin{smallmatrix}
\He(S_1(Y))&S_2(Y)+\mathcal{P}&S_3(Y)&S_4(Y)\\
\bullet&\mathcal{N}_{T_2}+\He(S_5(Y))&S_6(Y)&S_7(Y)\\
\bullet&\bullet&-\gamma^2\Id&\0\\
\bullet&\bullet&\bullet&-\chi\Id
\end{smallmatrix}\right)<\0
\end{aligned}
\end{equation}
where
\label{eq:FinslerProof0}
\begin{equation}
\label{eq:P1hat}
\begin{aligned}
&\mathcal{P}=\Diag\{P_1,P_2\}, \mathcal{P}_{T_2}=\Diag\{P_1,P_2e^{\delta T_2}\}\\
&\mathcal{N}=\Diag\{\lambda_t P_1+C_p\tr C_p+\chi\ell^2 S\tr S,(-\delta+2\lambda_t) P_2\}\\
&\mathcal{N}_{T_2}=\Diag\{\lambda_t P_1+C_p\tr C_p+\chi\ell^2 S\tr S,(-\delta+2\lambda_t)e^{\delta T_2} P_2\}\\
&X=\left(\begin{smallmatrix}
X_1& X_2&X_3&X_4\\
X_5&X_6&X_7&X_8
\end{smallmatrix}\right)\quad Y=\left(\begin{smallmatrix}
Y_1& Y_2&Y_3&Y_4\\
Y_5&Y_6&Y_7&Y_8
\end{smallmatrix}\right)
\end{aligned}
\end{equation}
and for each 
$$
\begin{aligned}
\mathcal{X}\in\R^{n_z\times n_z}\times\R^{n_z\times n_y}\times\R^{n_z\times n_z}\times\R^{n_z\times n_y}\\
\times\R^{n_y\times n_z}\times\R^{n_y\times n_y}\times\R^{n_y\times n_z}\times\R^{n_y\times n_y}
\end{aligned}
$$
\begin{equation}
\label{eq:FinslerProof}
\begin{aligned}
&S_1(\mathcal{X})\!=\left(\begin{smallmatrix}
-\mathcal{X}_1+C\tr \mathcal{X}_5&-\mathcal{X}_2+C\tr \mathcal{X}_6\\
-\mathcal{X}_5&-\mathcal{X}_6
\end{smallmatrix}\right)\\
&S_2(\mathcal{X})\!=\!\left(\begin{smallmatrix}
\mathcal{X}_1\tr (A-LC)-\mathcal{X}_5\tr HC-\mathcal{X}_3+C\tr \mathcal{X}_7&-\mathcal{X}_4+ C\tr \mathcal{X}_8+\mathcal{X}_1\tr L+\mathcal{X}_5\tr H\\
\mathcal{X}_2\tr (A-LC)-\mathcal{X}_6\tr HC-\mathcal{X}_7&-\mathcal{X}_8+\mathcal{X}_2\tr L+\mathcal{X}_6\tr H
\end{smallmatrix}\right)\\
&S_3(\mathcal{X})\!=\left(\begin{smallmatrix}
\mathcal{X}_1\tr N\\
\mathcal{X}_2\tr N\\
\end{smallmatrix}\right)\quad S_4(\mathcal{X})\!=\left(\begin{smallmatrix}
\mathcal{X}_1\tr  B\\
\mathcal{X}_2\tr  B\\
\end{smallmatrix}\right)\\
&S_5(\mathcal{X})\!=\left(\begin{smallmatrix}
(A-LC)\tr \mathcal{X}_3-C\tr H\tr \mathcal{X}_7&(A-LC)\tr \mathcal{X}_4-C\tr H\tr \mathcal{X}_8\\
L\tr \mathcal{X}_3+H\tr \mathcal{X}_7&L\tr \mathcal{X}_4+H\tr \mathcal{X}_8
\end{smallmatrix}\right)\\
&S_6(\mathcal{X})\!=\left(\begin{smallmatrix}
\mathcal{X}_3\tr N\\
\mathcal{X}_4\tr N\\
\end{smallmatrix}\right)\quad S_7(\mathcal{X})\!=\left(\begin{smallmatrix}
\mathcal{X}_3\tr  B\\
\mathcal{X}_4\tr  B\\
\end{smallmatrix}\right)
\end{aligned}
\end{equation}
Moreover, if $\delta>2\lambda_t$, then \eqref{eq:ProjeCoro} is fulfilled with $X_4=Y_4=\0,X_8=Y_8=\0$.
\end{theorem}
\begin{proof}
Let us define 
$$
{\footnotesize \mathcal{B}=\begin{pmatrix}
\mathcal{F}&\mathcal{T}&\mathcal{Q}\\
\Id&\0&\0\\
\0&\Id&\0\\
\0&\0&\Id
\end{pmatrix}}$$
where $\mathcal{F}$ and $\mathcal{T}$ are defined in \eqref{eq:P2:Chap3:FG}.
Then, $\mathcal{M}$ in \eqref{eq:M1} at $\tau=0$ and $\tau=T_2$ can be equivalently rewritten, respectively, as follows:
\begin{equation}
\label{eq:M1M2Projection}
\begin{aligned}
&\mathcal{M}(0)={\footnotesize \mathcal{B}\tr\underbrace{\left(\begin{array}{cc|cc}
	\0&\mathcal{P}&\0&\0\\
	\bullet&\mathcal{N}&\0&\0\\
	\hline
	\bullet&\bullet&-\gamma^2\Id&\0\\
	\bullet&\bullet&\bullet&-\chi\Id
	\end{array}\right)}_{\mathcal{Q}_1}\mathcal{B}}\\
&\mathcal{M}(T_2)={\footnotesize\mathcal{B}\tr\underbrace{\left(\begin{array}{cc|cc}
\0&\mathcal{P}_{T_2}&\0&\0\\
\bullet&\mathcal{N}_{T_2}&\0&\0\\
\hline
\bullet&\bullet&-\gamma^2\Id&\0\\
\bullet&\bullet&\bullet&-\chi\Id
\end{array}\right)}_{\mathcal{Q}_2}\mathcal{B}}
\end{aligned}
\end{equation}
Moreover, by defining 
\begin{equation}
\label{eq:U}
\mathcal{U}=\begin{pmatrix}\0_{2(n_z+n_y)\times (n_w+n_s)}\\
\Id_{n_w+n_s}
\end{pmatrix}
\end{equation}
it turns out that item $(i)$ of our assertion is equivalent to
\begin{equation}
\label{eq:P2Proj}
\left\{\begin{array}{ll}
\mathcal{U}\tr\mathcal{Q}_1\mathcal{U}<\0& \mathcal{B}\tr\mathcal{Q}_1\mathcal{B}<\0\\
\mathcal{U}\tr\mathcal{Q}_2\mathcal{U}<\0& \mathcal{B}\tr\mathcal{Q}_2\mathcal{B}<\0
\end{array}\right.
\end{equation}
Moreover, by the projection lemma; (see \cite{gahinet1994linear}) \eqref{eq:P2Proj} holds iff there exist two matrices $X,Y$ such that
\begin{equation}
\label{eq:Proje1}
\left\lbrace\begin{aligned}
&\mathcal{Q}_1+\mathcal{B}_r^{\perp\tr}X\mathcal{U}_r^{\perp}+{\mathcal{U}}_r^{\perp\tr}X\tr\mathcal{B}_r^{\perp}<\0\\
&\mathcal{Q}_2+\mathcal{B}_r^{\perp\tr} Y\mathcal{U}_r^{\perp}+{\mathcal{U}}_r^{\perp\tr}Y\tr\mathcal{B}_r^{\perp}<\0
\end{aligned}\right.
\end{equation}
where $\mathcal{B}_r^{\perp}$ and $\mathcal{U}_r^{\perp}$ are some matrices such that $\mathcal{B}_r^{\perp}\mathcal{B}=0$ and $\mathcal{U}_r^{\perp}\mathcal{U}=0$.  Specifically, notice that in view of Fact \ref{Claim}, one can consider the following choice
\begin{equation*}
\scriptsize{\begin{split}
&
\mathcal{B}_r^{\perp}=\begin{pmatrix}
-\mathcal{F}_l^{-1}&\mathcal{F}_r&\mathcal{F}_l^{-1}\mathcal{T}&\mathcal{F}_l^{-1}\mathcal{Q}
\end{pmatrix}=\\
&\left(\begin{array}{cc|cc|c|c} -\Id&\0 & A-LC&L& N& B\\ 
C&-\Id&-HC&H&\0&\0\end{array}\right)
\end{split}}
\end{equation*}
while $\mathcal{U}_r^{\perp}=\begin{pmatrix}
\Id_{2(n_z+n_y)}&\0_{2(n_z+n_y)\times (n_w+n_s)}
\end{pmatrix}$.
Thus, according to partitioning of $X$ and $Y$ in \eqref{eq:P1hat},
relation \eqref{eq:Proje1} turns into \eqref{eq:ProjeCoro}, hence $(i)\iff (ii)$. To conclude the proof, we need to show that whenever $\delta>2\lambda_t$ one has that \eqref{eq:ProjeCoro} is fulfilled iff $X_4=Y_4=\0$ and $X_8=Y_8=\0$. Define 
$\mathcal{U}_2=\begin{pmatrix}\0_{(2 n_z+n_y)\times (n_y+n_w+n_s)}\\
\Id_{n_y+n_w+n_s}
\end{pmatrix}$ and observe that if $\delta>2\lambda_t$, then $\mathcal{U}$ can be replaced by $\mathcal{U}_2$ in \eqref{eq:P2Proj}. Hence, still according to the projection lemma, $(i)$ is equivalent to the satisfaction of
\vspace{-0.65cm}

\begin{equation}
\label{eq:Proj2}
\left\lbrace\begin{aligned}
&\mathcal{Q}_1+\mathcal{B}_r^{\perp\tr}X\mathcal{U}_{2r}^{\perp}+{\mathcal{U}}_{2r}^{\perp\tr}X\tr\mathcal{B}_r^{\perp}<\0\\
&\mathcal{Q}_2+\mathcal{B}_r^{\perp\tr} Y\mathcal{U}_{2r}^{\perp}+{\mathcal{U}}_{2r}^{\perp\tr}Y\tr\mathcal{B}_r^{\perp}<\0
\end{aligned}\right.
\end{equation}
for some matrices $X,Y$. Hence, by noticing that $\mathcal{U}_2^{\perp}=\begin{pmatrix}
\Id_{2 n_z+n_y}&\0_{(2 n_z+n_y)\times (n_w+n_y+n_s)}
\end{pmatrix}$ and by considering the partitioning of $X,Y$ in \eqref{eq:P1hat}, it can be easily shown that \eqref{eq:Proj2} turns into \eqref{eq:ProjeCoro} with $X_4=Y_4=\0$ and $X_8=Y_8=\0$, hence finishing the proof.
\end{proof}
The above result yields an equivalent condition to \eqref{eq:M} that can be exploited to derive an efficient design procedure for the proposed observer. To this end, one needs to suitably manipulate \eqref{eq:ProjeCoro} to obtain conditions that are linear in the decision variables. Specifically, the three results given in the next section provide several possible approaches to derive sufficient conditions that whenever $\delta$ is selected are genuinely linear matrix inequalities. For the sake of brevity, we focus only on the exploitation of Theorem~\ref{Theorem1} with the aim of deriving sufficient conditions for the solution to Problem~\ref{prob:Problem1} in its whole.  Analogous arguments can be considered for Corollary~\ref{Coro:WoGamma}  and Corollary~\ref{Coro:WoLambda}. 
\begin{proposition}
\label{prop:propX80}
Let $\lambda_t, \gamma$ be given positive real numbers. If there exist $P_1\in\Spnz,P_2\in\Spny$, positive real numbers $\delta, \chi$, matrices $X\in\mathbb{R}^{n_z\times n_z},U, W\in\mathbb{R}^{n_y\times n_y}, J\in\mathbb{R}^{n_z\times n_y}$ such that
\begin{equation}
\label{eq:Design3}
\begin{aligned}
&\left(\begin{smallmatrix}
\He(Z_1)&Z_2+\mathcal{P}&Z_3&Z_4\\
\bullet&\mathcal{N}+\He(Z_5)&Z_6&Z_7\\
\bullet&\bullet&-\gamma^2\Id&\0\\
\bullet&\bullet&\bullet&-\chi\Id
\end{smallmatrix}\right)<\0\\
&\left(\begin{smallmatrix}
\He(Z_1)&Z_2+\mathcal{P}_{T_2}&Z_3&Z_4\\
\bullet&\mathcal{N}_{T_2}+\He(Z_5)&Z_6&Z_7\\
\bullet&\bullet&-\gamma^2\Id&\0\\
\bullet&\bullet&\bullet&-\chi\Id
\end{smallmatrix}\right)<\0
\end{aligned}
\end{equation}
where $\mathcal{P},\mathcal{P}_{T_2},\mathcal{N},\mathcal{N}_{T_2}$ are defined in \eqref{eq:P1hat} and 
$$
\scriptsize{\begin{array}{ll}
Z_1=\begin{pmatrix}
-X&C\tr U\\
\0&-U
\end{pmatrix}\hspace{-0.1cm}&Z_2=\begin{pmatrix}
-X+X\tr A-JC&J\\
-WC&W
\end{pmatrix}\\
Z_3=\begin{pmatrix}
X\tr N\\
\0
\end{pmatrix}\hspace{-0.32cm}& Z_4=\begin{pmatrix}
X\tr  B\\
\0
\end{pmatrix}\\
Z_5=\begin{pmatrix}
A\tr X-C\tr J\tr&\0\\
J\tr&\0
\end{pmatrix}\hspace{-0.32cm}&
Z_6=\begin{pmatrix}
X\tr N\\
\0
\end{pmatrix}
Z_7=\begin{pmatrix}
X\tr  B\\
\0
\end{pmatrix}
\end{array}}
$$
then $L=X\mtr J$ and $H=U\mtr W$ solve Problem~\ref{prob:Problem1}.
\end{proposition}
\begin{proof}
By selecting in \eqref{eq:ProjeCoro} $X_1=X_3=Y_1=Y_3=X, X_2=Y_2=\0,X_4=Y_4=\0, X_5=Y_5=\0, X_6=Y_6=U, X_7=Y_7=\0, X_8=Y_8=\0, X\tr L=J, U\tr H=W$ one gets \eqref{eq:Design3}. Thus, thanks to Theorem~\ref{Theorem1} and Theorem~\ref{Theorem:Projection} the result is proven.
\end{proof}
In Proposition~\ref{prop:propX80}, to obtain sufficient conditions in the form of (quasi)-LMIs, the following constraint is enforced $X_8=Y_8=\0$. Although this allows to obtain numerically tractable conditions, enforcing such a constraint, for a given $\lambda_t$, restricts the range of values of $\delta$ for which feasibility is not lost. Indeed, whenever  $X_8=Y_8=\0$, \eqref{eq:Design3} is feasible only if $-\delta+2\lambda_t<0$; due to the null lower-right corner block in $Z_5$. To overcome this obstacle, next we provide an additional result in which this limitation is removed.
\begin{proposition}
\label{prop:propX8X6}
Let $\lambda_t, \gamma$ be given positive real numbers. If there exist $P_1\in\Spnz,P_2\in\Spny$, positive real numbers $\delta, \chi$, matrices $X\in\mathbb{R}^{n_z\times n_z},U, W\in\mathbb{R}^{n_y\times n_y}, J\in\mathbb{R}^{n_z\times n_y}$ such that
\begin{equation}
\label{eq:Design3.2}
\begin{aligned}
&\left(\begin{smallmatrix}
\He(R_1)&R_2+\mathcal{P}&R_3&R_4\\
\bullet&\mathcal{N}+\He(R_5)&R_6&R_7\\
\bullet&\bullet&-\gamma^2\Id&\0\\
\bullet&\bullet&\bullet&-\chi\Id
\end{smallmatrix}\right)<\0\\
&\left(\begin{smallmatrix}
\He(R_1)&R_2+\mathcal{P}_{T_2}&R_3&R_4\\
\bullet&\mathcal{N}_{T_2}+\He(R_5)&R_6&R_7\\
\bullet&\bullet&-\gamma^2\Id&\0\\
\bullet&\bullet&\bullet&-\chi\Id
\end{smallmatrix}\right)<\0
\end{aligned}
\end{equation}
where $\mathcal{P},\mathcal{P}_{T_2},\mathcal{N},\mathcal{N}_{T_2}$ are defined in \eqref{eq:P1hat} and 
$$
\scriptsize{\begin{array}{ll}
R_1=\begin{pmatrix}
-X&C\tr U\\
\0&-U
\end{pmatrix}\\
R_2=\begin{pmatrix}
-X+X\tr A-JC&J+C\tr U\\
-WC&-U+W
\end{pmatrix}\\
R_3=\begin{pmatrix}
X\tr N\\
\0
\end{pmatrix}&R_4=\begin{pmatrix}
X\tr  B\\
\0
\end{pmatrix}\\ 
R_5=\begin{pmatrix}
A\tr X-C\tr J\tr&-C\tr W\tr\\
J\tr&W\tr
\end{pmatrix}&\hspace{-0.0cm}
R_6=\begin{pmatrix}
X\tr N\\
\0
\end{pmatrix}\\
R_7=\begin{pmatrix}
X\tr  B\\
\0
\end{pmatrix}
\end{array}}
$$
then $L=X\mtr J$ and $H=U\mtr W$ solve Problem~\ref{prob:Problem1}.
\end{proposition}
\begin{proof}
By selecting in \eqref{eq:ProjeCoro} $X_1=X_3=Y_1=Y_3=X, X_2=Y_2=\0,X_4=Y_4=\0, X_5=Y_5=\0, X_6=Y_6=X_8=Y_8=U, X_7=Y_7=\0, X\tr L=J, U\tr H=W$, one gets \eqref{eq:Design3}. Thus, thanks to Theorem~\ref{Theorem1} and Theorem~\ref{Theorem:Projection}  the result is proven.
\end{proof}
\begin{remark}
As already mentioned, the above result, with respect to Proposition~\ref{prop:propX80}, extends the range of values for $\delta$ for which feasibility is not lost. However, it is difficult to compare the conservatism induced by Proposition~\ref{prop:propX80} and Proposition~\ref{prop:propX8X6}. Therefore, in practice the two above results need to be used in a complementary fashion. 
\end{remark}
\subsubsection*{\textbf{Sample-and-hold Implementation}}
Whenever $H=\0$, the general observer scheme presented in this paper reduces to the \emph{zero order holder (ZOH) sample-and-hold} considered, \eg, in \cite{raff2008observer}. Although such an observer is perfectly captured by our scheme, the implementation of ZOH sample-and-hold observer schemes only requires to store the last measured output estimation error and hold it in between sampling times. Thus, implementing such schemes is in general easier.
For this reason, it appears useful to derive computationally tractable design algorithms for which the gain $H$ is explicitly constrained to be zero. This is realized through the following result. 
\begin{proposition}
\label{prop:ZOH}
Let $\lambda_t, \gamma$ be given positive real numbers. If there exist $P_1\in\Spnz,P_2\in\Spny$, positive real numbers $\delta, \chi$, a nonsingular matrix $X\in\mathbb{R}^{n_z\times n_z}$, and matrices $X_5,Y_5,X_7,Y_7\in\R^{n_y\times n_z},X_6,Y_6,X_8,Y_8\in\mathbb{R}^{n_y\times n_y}, J\in\mathbb{R}^{n_z\times n_y}$ such that
\begin{equation}
\label{eq:DesignZOH}
\begin{aligned}
&\left(\begin{smallmatrix}
\He(Q_1)&Q_2+\mathcal{P}&Q_3&Q_4\\
\bullet&\mathcal{N}+\He(Q_5)&Q_6&Q_7\\
\bullet&\bullet&-\gamma^2\Id&\0\\
\bullet&\bullet&\bullet&-\chi\Id
\end{smallmatrix}\right)<\0\\
&\left(\begin{smallmatrix}
\He(\widehat{Q}_1)&\widehat{Q}_2+\mathcal{P}_{T_2}&Q_3&Q_4\\
\bullet&\mathcal{N}_{T_2}+\He(Q_5)&Q_6&Q_7\\
\bullet&\bullet&-\gamma^2\Id&\0\\
\bullet&\bullet&\bullet&-\chi\Id
\end{smallmatrix}\right)<\0\\
\end{aligned}
\end{equation}
where $\mathcal{P},\mathcal{P}_{T_2},\mathcal{N},\mathcal{N}_{T_2}$ are defined in \eqref{eq:P1hat} and 
$$
\scriptsize{\begin{array}{ll}
Q_1=\begin{pmatrix}
-X+C\tr X_5&C\tr X_6\\
-X_5&-X_6
\end{pmatrix}&\\
Q_2=\begin{pmatrix}
-X+X\tr A-JC+C\tr X_7&J+C\tr X_8\\
-X_7&-X_8
\end{pmatrix}\\
Q_3=\begin{pmatrix}
X\tr N\\
\0
\end{pmatrix}&\hspace{-2cm}Q_4=\begin{pmatrix}
X\tr  B\\
\0
\end{pmatrix}\\
Q_5=\begin{pmatrix}
A\tr X-C\tr J\tr&\0\\
J\tr&\0
\end{pmatrix}&\hspace{-2cm}Q_6=\begin{pmatrix}
X\tr N\\
\0
\end{pmatrix}\\
Q_7=\begin{pmatrix}
X\tr  B\\
\0
\end{pmatrix}\quad
\widehat{Q}_1=\begin{pmatrix}
-X+C\tr Y_5&C\tr Y_6\\
-Y_5&-Y_6
\end{pmatrix}\\
\widehat{Q}_2=\begin{pmatrix}
-X+X\tr A-JC+C\tr Y_7&J+C\tr Y_8\\
-Y_7&-Y_8
\end{pmatrix}
\end{array}}
$$
then $L=X\mtr J$ and $H=\0$ are a solution to Problem~\ref{prob:Problem1}.
\end{proposition}
\begin{proof}
By selecting in \eqref{eq:ProjeCoro} $H=\0,X_1=X_3=Y_1=Y_3=X,X_2=Y_2=\0,X_4=Y_4=\0,X\tr L=J$, one gets \eqref{eq:DesignZOH}. Thus, thanks to Theorem~\ref{Theorem1} and Theorem~\ref{Theorem:Projection} the result is proven.
\end{proof}
\begin{remark}
The applicability of the above result requires the matrix $X$ to be nonsingular and such a constraint cannot be directly imposed in an LMI setting. 
Although into a solution to \eqref{eq:DesignZOH} characterized by a singular matrix $X$ is unlikely, if one wants to ensure the nonsingularity of $X$, at the expense of some additional conservatism, then the following constraint can be included $X\tr+X>\0$.
\end{remark}
\begin{remark}
The proposed design procedures lead to a different number of scalar variables in the associated LMIs. Table~\ref{table:Complexity} reports such a number for each of the proposed design. As it appears from the table, the matrix inequalities related to Proposition~\ref{prop:propX80} (or equivalently Proposition~\ref{prop:propX8X6}) and Proposition~\ref{prop:ZOH}, due to the introduction of the additional slack variables, lead to a greater number of scalar variables with respect to the matrix inequalities issued from Proposition~\ref{prop:P2:Chap3:propPred}. \figurename~\ref{fig:NVar} reports the number of scalar variables associated to the different results as a function of $n_z$ whenever $n_y=1$. The picture clearly points out that design algorithms based on Proposition~\ref{prop:P2:Chap3:propPred} are more preferable when the plant order is sufficiently large.   
\end{remark}
{\small\begin{table}[h]
\centering
\begin{tabular}{ll}
\toprule
Design& \# scalar variables\\ [0.1ex] 
\midrule
Prop.~\ref{prop:P2:Chap3:propPred}&$n_z(n_z+1)/2+n_y (n_y+1)/2+n_y^2+n_z n_y+1$\\	
Prop.~\ref{prop:propX80}&$n_z(n_z+1)/2+n_y(n_y+1)/2+2 n_y^2+n_z^2+n_z n_y+1$\\	
Prop.~\ref{prop:ZOH}&$n_z(n_z+1)/2+n_y(n_y+1)/2+4 n_y^2+n_z^2+5 n_z n_y+1$	\\
\hline
\end{tabular}
\caption{Number of scalar variables associated to the different designs.}
\label{table:Complexity}
\end{table}}
\begin{figure}[!ht]
\centering
\psfrag{nz}[][][1]{$n_z$}
\psfrag{nv}[][][1]{$\#\,\,\text{scalar var.}$}
 \includegraphics[width=0.5\textwidth]{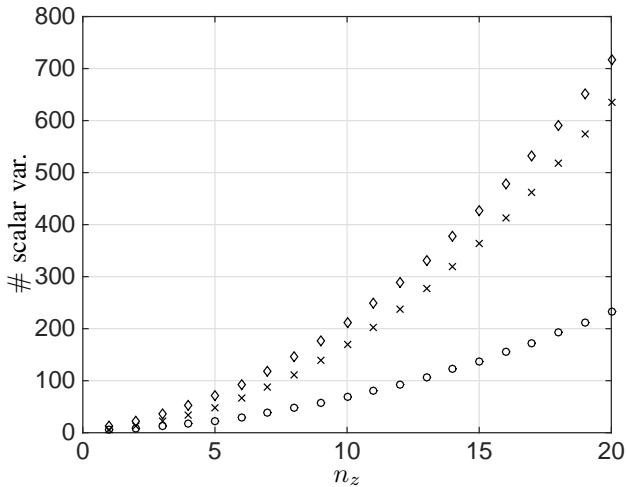}
\caption{Number of scalar variables associated to the different results vs $n_z$ whenever $n_y=1$: Proposition~\ref{prop:P2:Chap3:propPred} (circle), Proposition~\ref{prop:propX80} (cross), Proposition~\ref{prop:ZOH} (diamond).}
\label{fig:NVar}
\end{figure}
\subsection{Optimization aspects}
So far, we assumed the gain $\gamma$ to be given. Nonetheless, most of the time one is interested in designing the observer in a way such that the effect of the exogenous signals is reduced as much as possible. 
This can be realized in our setting by embedding the proposed design conditions into suitable optimization schemes aimed at minimizing  $\gamma$,  which can be taken as a design variable. In particular, by setting $\gamma^2=\mu$, the minimization of the $\mathcal{L}_2$ gain from the disturbance $w$ to the performance output $y_p$ can be achieved, for a given value of $\lambda_t>0$, by designing the observer via the solution to the following optimization problem:
\begin{equation}
\label{eq:OptSingle}
\begin{aligned}
&\underset{P_1,P_2,L, H,\mu, \delta,\chi}{\minimize}\,\,\,\mu\\
&\text{s.t.}\\
&P_1\in\Spnz, P_2\in\Spny, \mu>0, \delta>0, \chi\geq 0\\
&\mathcal{M}(0)\leq \0, \mathcal{M}(T_2)\leq \0
\end{aligned}
\end{equation}
Clearly the above optimization problem is hardly tractable from a numerical standpoint due to nonlinear constraints in the decision variables. However, whenever $\delta$ is given, the results given in Section~\ref{sec:LMI} allows to obtain sufficient conditions in the form of linear matrix inequalities for the satisfaction of \eqref{eq:M}. Thus,  a suboptimal solution to the above optimization problem can be obtained via semidefinite programming by performing a grid search for the scalar $\delta$. 
\begin{remark}
\label{rem:TwoStage}
The derivation of Proposition~\ref{prop:propX80}, Proposition~\ref{prop:propX8X6}, and Proposition~\ref{prop:ZOH} consists of some particular choices of the slack variables $X$ and $Y$ introduced in Theorem~\ref{Theorem:Projection}. Therefore, the adoption of such results for the derivation of suboptimal solutions to \eqref{eq:OptSingle} may prevent from solving Problem~\ref{prob:Problem1} for a given value of  $T_2$. 
To overcome this problem, one can envision a two-stage procedure.  Indeed, whenever $L$, $H$, $\delta$ and $T_2$ are fixed,  $\mathcal{M}(0)\leq \0, \mathcal{M}(T_2)\leq \0$ are linear in the decision variables. Thus, once the observer has been designed, by testing the feasibility of $\mathcal{M}(0)\leq \0, \mathcal{M}(T_2)\leq \0$ with respect to $P_1,P_2$ over a selected grid for the variables $\delta$ and $T_2$, one may be able to enlarge the maximum allowable transfer interval $T_2$ as well as to get a smaller value of $\gamma$.
\end{remark}
The maximum transfer time $T_2$ can be considered as a design parameter within an optimization scheme as the one outlined above. Indeed, with the aim of decreasing the amount of information needed to reconstruct the plant state, one may be interested, for some given positive values of $\gamma$ and $\lambda_t$, in designing the observer gains while maximizing the allowable value of $T_2$, that is the value of $T_2$ for which Problem~\ref{prob:Problem1} is feasible.
This can be accomplished by solving the following optimization problem:
\begin{equation}
\label{eq:OptTime}
\begin{aligned}
&\underset{P_1,P_2,L, H,\delta, T_2}{\minimize}\,\,\,-T_2\\
&\text{s.t.}\\
&P_1\in\Spnz, P_2\in\Spny, \mu>0, \delta>0, T_2>0, \chi\geq 0\\
&\mathcal{M}(0)\leq \0, \mathcal{M}(T_2)\leq \0
\end{aligned}
\end{equation}
Also in this case, the above optimization problem is difficult to solve in practice due to nonlinear  matrix inequality constraints. On the other hand, being the objective function of \eqref{eq:OptTime} linear in the decision variables,  thanks to the results given in Section~\ref{sec:LMI}, the above optimization problem can be solved (suboptimally) via semidefinite programming along with a bisection algorithm (see, \eg,\cite{Boyd}), with the only caveat of performing a grid search for the variable $\delta$.

Whenever one is interested in achieving both objectives simultaneously, the two above optimization problems can be blended together to give rise, for a given value of $\lambda_t>0$, to the following multiobjective optimization problem.
\begin{equation}
\label{eq:OptMulti}
\begin{aligned}
&\underset{P_1,P_2,L, H,\mu, \delta, T_2, \chi}{\minimize}(\text{w.r.t.}\,\,\R_{\geq 0}\times\R_{\geq 0})\,\,\,(-T_2,\mu)\\
&\text{s.t.}\\
&P_1\in\Spnz, P_2\in\Spny,\mu>0, \delta>0, \chi\geq 0\\
&\mathcal{M}(0)\leq \0, \mathcal{M}(T_2)\leq \0
\end{aligned}
\end{equation}
where $\minimize(\text{w.r.t.}\,\,\R_{\geq 0}\times\R_{\geq 0})$ stands for the componentwise minimum in $\R^2$ \cite{boyd2004convex}. 
An effective method used in practice to get ``good feasible points'' out of a (bidimensional) optimization problem consists of visualizing (an approximation of) the (Pareto) tradeoff curve and then selecting the most convenient solution; see \cite{miettinen2012nonlinear}. 
An approximation of such a curve can be obtained in our setting by solving the optimization problem in \eqref{eq:OptSingle} over a range of values for $T_2$.
The main advantage is that,  by relying on either Proposition~\ref{prop:propX80}, Proposition~\ref{prop:propX8X6}, or Proposition~\ref{prop:ZOH},  a suboptimal solution to \eqref{eq:OptSingle} can be obtained via off-the-shelf semidefinite programming software. 
\begin{remark}
Depending on the application, one could need to either enforce a certain convergence speed or to limit the number of sampling events. In any case, to avoid penalizing too much the convergence properties of the observer,
a suitable trade-off between this two antagonistic objectives needs to be considered. One of the strengths of our methodology is that (as for the case of the pair $(T_2,\gamma)$) it allows the designer to systematically build an approximation of trade-off curve for the objective $(T_2, \lambda_t)$. 
\end{remark}
\section{Numerical examples}
In this section, we showcase the effectiveness of our methodology in three examples. The first example is academic and pertains to the linear oscillator in \cite{karafyllis2009continuous}, for which we show how our suboptimal design allows to improve disturbance rejection and convergence speed. The second example is of practical interest and pertains to the path following unicycle robot in \cite{moarref2014observer}. In this example, we show how our design methodology allows for the design of a sample-and-hold observer and how this compares with other results in the literature. Finally, the third example is also of practical interest and pertains to a widely studied nonlinear plant in the context of observer design, \ie,  the flexible one-link manipulator \cite{howell2002nonlinear,spong1987modeling}. In this example, we show how the different deigns we propose compare each other and how the methodology presented in this paper leads to less conservative results when compared to existing approaches.

Numerical solutions to the semidefinite programming problems arising in the examples are obtained through the solver \emph{SDPT3} \cite{tutuncu2003solving} and coded in Matlab$^{\textregistered}$ via \emph{YALMIP} \cite{lofberg2004yalmip}. Simulations of hybrid systems are performed in Matlab$^{\textregistered}$ via the  
\emph{Hybrid Equations (HyEQ) Toolbox} \cite{sanfelice2013toolbox}.
\label{sec:Examples}
\begin{example}
\label{Example1}
In this first example, we want to show the improvement provided by our methodology with respect to existing results.  
Specifically, consider the example in \cite{0801.4824}, which is defined by the following data: 
$A=\begin{pmatrix}
0&1\\
-4&0
\end{pmatrix}, C=\begin{pmatrix}
1&0\end{pmatrix}$
as a performance output, we pick $C_p=\Id$ and as input matrix we select $N=\begin{pmatrix}
1&0\end{pmatrix}\tr$.
We solve the multi-objective optimization problem \eqref{eq:OptMulti} with $\lambda_t=0.05$. As already mentioned,  the suboptimal solution to such a problem can be obtained in a different way, depending on which result is exploited to solve the underlying single objective optimization problem \eqref{eq:OptSingle}. To give a complete panorama of our methodology, in \figurename~\ref{fig:MultiObj} we show the resulting tradeoff curve for each of the proposed results.
\begin{figure}[!htbp]
\centering
\psfrag{T}[][][1.1]{$T_2$}
\psfrag{g}[][][1.1]{$\gamma$}
 \includegraphics[width=0.5\textwidth]{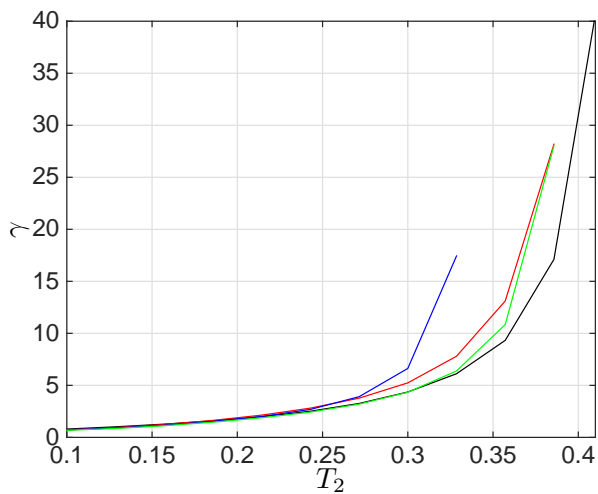}
\caption{Tradeoff curves obtained by considering different relaxations: Proposition~\ref{prop:P2:Chap3:propPred} (black), Proposition~\ref{prop:ZOH} (red), Proposition~\ref{prop:propX8X6} (green), and Proposition~\ref{prop:propX80} (blue).}
\label{fig:MultiObj}
\end{figure}
\figurename~\ref{fig:MultiObj} points out that, in this example, Proposition~\ref{prop:P2:Chap3:propPred} gives the best result overall. 

In \cite{ferrante2015hybrid}, it is shown that for $T_2=0.41$, the pair
\begin{equation}
L\tr=\begin{pmatrix}
0.3648&-0.4655
\end{pmatrix}, H=-CL=-0.3648
\label{eq:OldGains}
\end{equation}
provides a solution to item $(P1)$ in Problem~\ref{prob:Problem1}, for $\lambda_t$ small enough.
On the other hand, \figurename~\ref{fig:MultiObj} shows that $T_2=0.41$ corresponds to a feasible solution to \eqref{eq:OptMulti}, when one relies on Proposition~\ref{prop:P2:Chap3:propPred} as a design result. In particular, the Pareto (sub) optimal solution associated to such a value of $T_2$ is characterized by the following data\footnote{To avoid the occurrence of an overly large norm for the gain $L$, which would give rise to numerical and implementation issues, in the solution to \eqref{eq:OptSingle} we considered a further constraint aimed at limiting the norm of $L$.}
\begin{equation}
L=\begin{pmatrix}
 2.067\\
-3
\end{pmatrix}, H=-1.384, \gamma=36
\label{eq:NewGains}
\end{equation}
To show the effectiveness of the proposed suboptimal design,  in \figurename~\ref{fig:Comprison_Disturbance}, we compare two solutions $\phi^a=(\phi^a_{z},\phi^a_{\varepsilon}, \phi^a_{\thetatilde}, \phi^a_{\tau})$ and $\phi^b=(\phi^b_{z},\phi^b_{\varepsilon}, \phi^b_{\thetatilde}, \phi^b_{\tau})$ to $\mathcal{H}_e$, obtained, respectively, for the suboptimal gains in \eqref{eq:NewGains} and for the gains in \eqref{eq:OldGains} from zero initial conditions in response to the following exogenous input $\tilde{w}\in\mathcal{L}_2$
$$
\tilde{w}(t)=\begin{cases}
-1& t\in[0,5]\\
1& t\in(5,10]\\
-1& t\in(10,15]\\
0& t>15
\end{cases}
$$
In this simulation, $T_1=0.5 T_2$, $\phi^a(0,0)=\phi^b(0,0)=(0,0,0, 0, 0, T_2)$, 
$\phi^a_{\tau}=\phi^b_{\tau}\coloneqq\phi_{\tau}$, and for each 
$(t_j, j+1)\in\dom\phi_a=\dom\phi_b$
\begin{equation}
\label{eq:sequence_gen}
{\footnotesize\begin{aligned}
\phi_{\tau}(t_j, j+1)=\frac{T_2-T_1}{2}\sin(10 t_j)+\frac{T_2+T_1}{2}
\end{aligned}}
\end{equation}
\begin{figure}[!htbp]
\centering
\psfrag{t}[][][1.2]{$t$}
\psfrag{ep1a,ep1b}[][][1.2]{$\ep_1$}
\psfrag{ep2a,ep2b}[][][1.2]{$\ep_2$}
\psfrag{w}[][][1.2]{$\tilde{w}$}
 \includegraphics[width=0.5\textwidth]{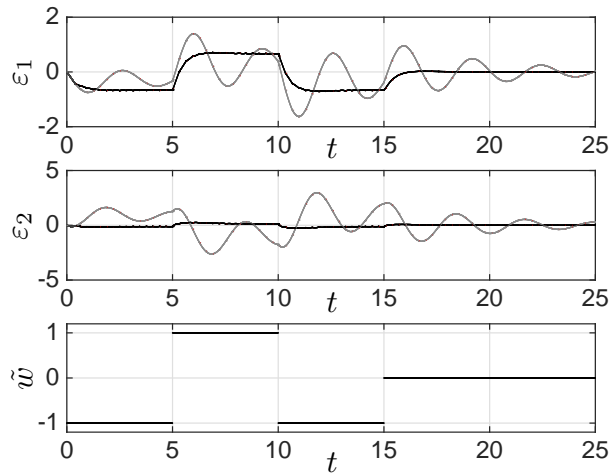}
\caption{Evolution of $\phi^{a}$ (black) and of $\phi^{b}$ (gray) projected onto ordinary time, and $t\mapsto \tilde{w}$.}
\label{fig:Comprison_Disturbance}
\end{figure}
Simulations show that the proposed suboptimal design leads to better performances in terms of rejection of the exogenous perturbation. 
\begin{figure}[!htbp]
\centering
\psfrag{t}[][][1.2]{$t$}
\psfrag{ep1a,ep1b}[][][1.2]{$\ep_1$}
\psfrag{ep2a,ep2b}[][][1.2]{$\ep_2$}
\psfrag{tha,thb}[][][1.2]{$\thetatilde$}
\includegraphics[width=0.5\textwidth]{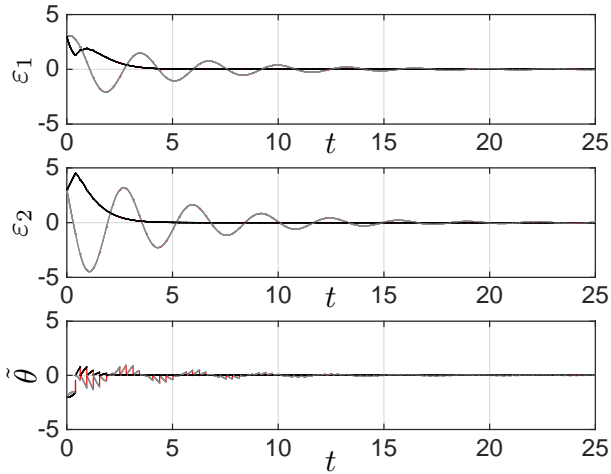}
\caption{Evolution of $\varphi^a$ (black) and 
$\varphi^b$ (gray) projected onto ordinary time.}
\label{fig:Comprison_DisturbanceFree}
\end{figure}
To analyze the convergence of the estimation error in absence of exogenous inputs, in \figurename~\ref{fig:Comprison_DisturbanceFree} we compare two solutions $\varphi^a=(\varphi^a_{z},\varphi^a_{\varepsilon}, \varphi^a_{\thetatilde}, \varphi^a_{\tau})$ and $\varphi^b=(\varphi^b_{z},\varphi^b_{\varepsilon}, \varphi^b_{\thetatilde}, \varphi^b_{\tau})$ to $\mathcal{H}_e$ obtained, respectively, with the gains in \eqref{eq:NewGains} and with the  gains in \eqref{eq:OldGains} and with $w\equiv 0, \eta\equiv 0$. In this simulation, $T_1=0.5 T_2$, $\varphi^a_{\tau}(0,0)=\varphi^b_{\tau}(0,0)=(1, 1,3,3, -2, T_2)$,  and $\varphi^a_{\tau}=\varphi^b_{\tau}\coloneqq\varphi_{\tau}$, where $\varphi_{\tau}$ satisfies \eqref{eq:sequence_gen}.
Simulations show that the proposed suboptimal design, thanks to specification of a certain $t$-decay rate, ensures also a faster convergence of the estimation error and of the error $\thetatilde$.
\end{example}
\begin{example}
Consider the linearized model of the path following unicycle robot in \cite{moarref2014observer}, that is defined as follows
\begin{equation}
\dot{z}=\underbrace{\begin{pmatrix} 
0 & 0 & 1\\ 
0 & -0.01 & 0\\ 1
1 & 0 & 0\end{pmatrix}}_{A}z+\underbrace{\begin{pmatrix} 0\\ 1\\0 \end{pmatrix}}_{N}w
\end{equation}
where $z_1$ is the distance of the robot to the target line, $z_2$ is the heading angle, $z_3$ is the yaw angular speed, and $w$ is an external torque. 
Assume that $z_1$ and $z_3$ can be measured with sampling time $T_1=0.1714s$ affected by an uncertain jitter $\Delta_{T_1}$. Namely, the vector $y(t)=(z_1(t),z_3(t))$ is measured only at certain time instances $t_k$, for $k\in\mathbb{N}_{>0}$, where the sequence $\{t_k\}_{k=1}^\infty$ fulfills \eqref{eq:P2:Chap3:timebound} with $T_2=T_1+\Delta_{T_1}$.
Under these assumptions, we want to design an observer providing an estimate $\hat{z}$ of the state $z$ for the largest allowable jitter $\Delta_{T_1}$ while minimizing the $\mathcal{L}_2$ gain from the exogenous input $w$ to the performance output $y_p=z_2-\hat{z}_2$. Moreover, to guarantee a certain performance in the convergence speed, we want to enforce a decay rate $\lambda_t=0.2$.
The considered problem can be put into the setting of Problem~\ref{prob:Problem1} by taking
$\begin{aligned}
&C_p=\begin{pmatrix}
0&1&0
\end{pmatrix},
&M=\begin{pmatrix}
1&0&0\\
0&0&1
\end{pmatrix}
\end{aligned}$.
Therefore, to achieve a tradeoff between robustness to sampling time jitter and disturbance rejection, we design the observer via the solution to the multi-objective optimization problem \eqref{eq:OptMulti}. In particular, to compare our results with more classical approaches based on sampled-data observers, we designed the observer via Proposition~\ref{prop:ZOH}, which enforces $H=0$, leading to the same observer in \cite{moarref2014observer}. 
The resulting tradeoff curve is depicted in \figurename~\ref{fig:TradeOffRobot}.
\begin{figure}[!htbp]
\centering
\centering
\psfrag{g}[][][1]{$\gamma$}
\psfrag{deltaT1}[][][1.2]{$\frac{\Delta_{T_1}}{T_1}$}
\includegraphics[width=0.5\textwidth]{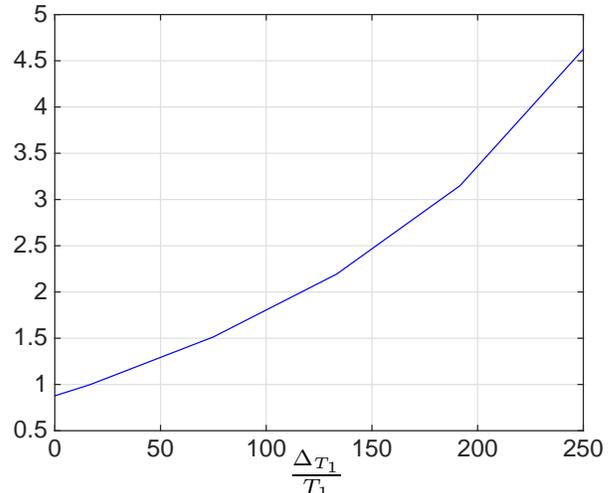}
\caption{Tradeoff curve versus the amplitude of the relative jitter in percentage.}
\label{fig:TradeOffRobot}
\end{figure}
\begin{figure}[!htbp]
	\centering
\psfrag{t}[][][1.2]{$t$}
\includegraphics[width=0.5\textwidth]{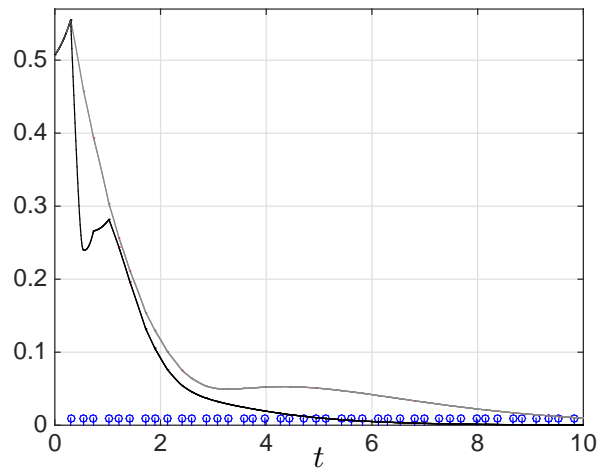}
\caption{Evolution of $\vert \varphi^a\vert_{\mathcal{A}}$ (black) and $\vert \varphi^b\vert_{\mathcal{A}}$ (gray) projected onto ordinary time. The blue bullets denote the sampling instances.}
\label{fig:FreeRobot}
\end{figure}
By selecting the tradeoff value $(T_2, \gamma)=(0.3, 1.5140)$ that corresponds to a relative jitter of $75\%$ with respect to the nominal sampling time, one gets 
$L=\left(\begin{smallmatrix}
3.7 & -2.194\\ 2.908 & -2.075\\ 1.637 & 0.1545
\end{smallmatrix}\right)$.
In \cite{moarref2014observer}, the authors show that the gain
$L_2=\left(\begin{smallmatrix}
0.8079  & 0.2555\\
0.2071  & 0.0550\\
0.7609  & 0.7714
\end{smallmatrix}\right)$
provides a solution to the considered estimation problem for $T_2=0.3$.
To show the effectiveness of the proposed suboptimal design,  in \figurename~\ref{fig:DisturbanceRobot}, we report two solutions $\phi^a=(\phi^a_{z},\phi^a_{\varepsilon}, \phi^a_{\thetatilde}, \phi^a_{\tau})$ and $\phi^b=(\phi^b_{z},\phi^b_{\varepsilon}, \phi^b_{\thetatilde}, \phi^b_{\tau})$ to $\mathcal{H}_e$ obtained, respectively, in correspondence to the gain $L$ and $L_2$,  from zero initial conditions in response to the following exogenous input $\tilde{w}\in\mathcal{L}_2$
$$
\tilde{w}(t)=\begin{cases}
1& t\in[0,2]\\
0& t\in(2,6]\\
-1& t\in(6,8]\\
0& t>8\\
\end{cases}
$$
Analogously to Example~\ref{Example1}, also in this simulations $\phi^a_{\tau}=\phi^b_{\tau}\coloneqq \phi_{\tau}$, where  $\phi_{\tau}$ satisfies  \eqref{eq:sequence_gen}.
\begin{figure}[!hbp]
\centering
\psfrag{t}[][][1]{$t$}
\psfrag{ep1}[][][1.2]{$\ep_1$}
\psfrag{ep2}[][][1.2]{$\ep_2$}
\psfrag{ep3}[][][1.2]{$\ep_3$}
 \includegraphics[trim={2cm 0 0.2cm 0},clip, width=0.52\textwidth]{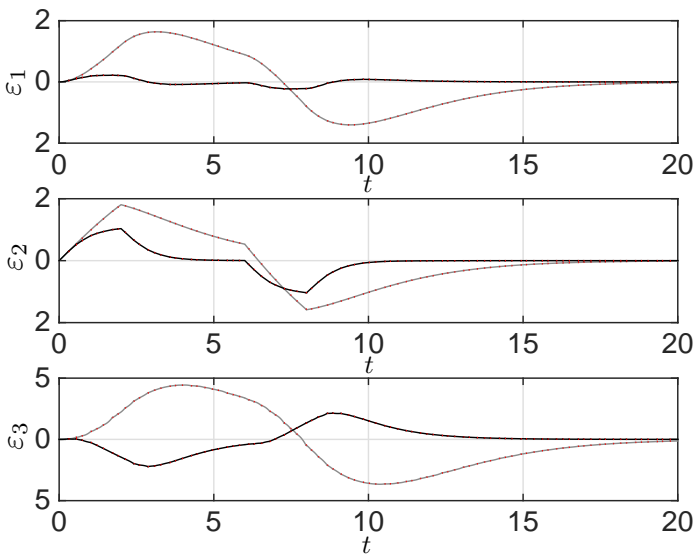}
\caption{Evolution of $\phi^a$ (black) and $\phi^b$ (gray) projected onto ordinary time.}
\label{fig:DisturbanceRobot}
\end{figure}

Simulations show that the proposed design provides better performance in terms of disturbance rejection. Finally, in \figurename~\ref{fig:FreeRobot}, we compare two solutions $\varphi^a$ and $\varphi^b$ to $\mathcal{H}_e$ obtained with $w\equiv 0$, respectively, with the gain $L$ and $L_2$. In this simulation, 
$\varphi^a(0,0)=\varphi^b=(0.5, 0.0873,0, 0.5, 0.0873,0.5,0,0, T_2)$ and, as in the former simulation, the $\tau$-component of both solutions coincide and satisfy  \eqref{eq:sequence_gen}.
Simulations point out that the design we propose not only provide improved disturbance rejection but also ensures a faster transient response with respect to a non-optimal design.    
\end{example}
\begin{example}
Consider the following model of the flexible one-link manipulator \cite{howell2002nonlinear,spong1987modeling}
$$
\begin{aligned}
&\dot{z}=\underbrace{\left(\begin{smallmatrix}0& 1 &0 &0\\
-48.6& -1.25& 48.6& 0\\
0 &0 &0 &1\\
19.5& 0 &-19.5& 0\end{smallmatrix}\right)}_{A}z+\left(\begin{smallmatrix}0\\ 0\\ 0\\ -3.33\end{smallmatrix}\right)\sin(z_3)+\underbrace{\left(\begin{smallmatrix}0\\ 2\\ 0\\ 0
\end{smallmatrix}\right)}_{N}w\\
&y=\underbrace{\left(\begin{smallmatrix}
1& 0 &0 &0\\
0 &1 &0 &0
\end{smallmatrix}\right)}_{C}z
\end{aligned}
$$
where $z_1$ and $z_2$ are, respectively, the motor shaft angle and the motor shaft angular speed, while $z_3$ and $z_4$ are, respectively, the link angle and the link angular speed. The exogenous input $w$ represents a disturbance torque acting on the motor shaft.
\begin{figure}[!htbp]
	\centering
	\psfrag{T2}[][][0.9]{$T_2$}
	\psfrag{gamma}[][][0.9]{$\gamma$}
	\includegraphics[width=0.5\textwidth]{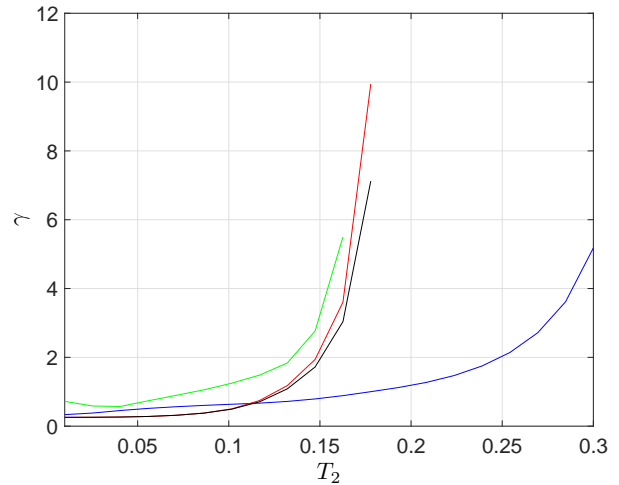}
	\caption{Tradeoff curves obtained by considering different relaxations: Proposition~\ref{prop:P2:Chap3:propPred} (blue), Proposition~\ref{prop:propX80} (red), Proposition~\ref{prop:propX8X6} (black), Proposition~\ref{prop:ZOH} (green).}
	\label{fig:MultiNonLinear}
\end{figure}
Assuming the output $y$ can be measured sporadically, we want to design an observer providing an estimate $\hat{z}$ of $z$ while reducing the effect of the exogenous signal $w$ on the estimate of the unmeasured link variables $z_3$ and $z_4$.
By setting $B=\left(\begin{smallmatrix}0 &0& 0&-1\end{smallmatrix}\right)\tr,  S=\left(\begin{smallmatrix}0& 0 &1 &0\end{smallmatrix}\right)$, $\ell=3.3$, $C_p=\left(\begin{smallmatrix}0& 0 &1 &0\\0& 0 &0 &1\end{smallmatrix}\right)$,
the considered plant can be rewritten as \eqref{eq:P2:Chap3:Plant}, so that the methodology proposed in the paper can be applied. 
Figure~\ref{fig:MultiNonLinear} shows the tradeoff curves associated with the multiobjective optimization \eqref{eq:OptMulti} obtained via the proposed SDP-based relaxation schemes issued from Proposition~\ref{prop:P2:Chap3:propPred}, Proposition~\ref{prop:propX80}, Proposition~\ref{prop:propX8X6}, and Proposition~\ref{prop:ZOH}; in this example $\lambda_t=0.01$, $\delta$ is selected over a grid of $100$ points spanning the interval $[1, 100]$, and $T_2$ is selected over a grid of $20$ points spanning the interval $[0.01, 0.3]$. 
To reduce the conservatism in the estimation on the $\mathcal{L}_2$ gain, in the case of Proposition~\ref{prop:ZOH} a further analysis stage based on Theorem~\ref{Theorem1} is included in the solution to  
multiobjective optimization \eqref{eq:OptMulti}. As pointed out earlier, each relaxation leads to a different number of scalar variables in the resulting LMIs, which in turn reflects on a different computational complexity. Table~\ref{table:ComplexityExample} reports the number of scalar variables and the computation time of the tradeoff curve for each relaxation scheme\footnote{For the case of Proposition~\ref{prop:ZOH}, the computation time includes the additional analysis stage. When such an analysis stage is not considered, the computation time decreases to $1192.226 s$, that is $6.9960 s$ smaller.}. Computations are performed on an \emph{iMac 3.2 GHz Intel Core i5 RAM 16 GB}.
{\small\begin{table}[h]
\centering
\begin{tabular}{lll}
\toprule
Design& \# scalar variables& Time [s]\\ [0.1ex] 
\midrule
Prop.~\ref{prop:P2:Chap3:propPred}&$26$&$716.642$\\	
Prop.~\ref{prop:propX80}&$46$&$866.417$\\	
Prop.~\ref{prop:ZOH}&$86$&$1199.222$\\

\hline
\end{tabular}
\caption{Number of scalar variables and computation time associated to the different designs.}
\label{table:ComplexityExample}
\end{table}}
In \cite{raff2008observer}, sufficient conditions in the form of LMIs are given for the design of a sample-and-hold observer that solves item $(P1)$ of Problem~\ref{prob:Problem1}. In particular for this example, the conditions given in \cite{raff2008observer} are feasible for $T_2$ up to $0.1$. Figure~\ref{fig:MultiNonLinear} shows that our methodology allows not only to guarantee robustness with respect to external inputs and $\mathcal{L}_2$-gain performance, but also leads to a larger allowable value for $T_2$. Specifically, $T_2$ can be selected up to $0.3$, \ie, an improvement of $200\%$ with respect to \cite{raff2008observer}.

With the aim of getting a good trade-off between the reduction of the effect of the external disturbance on the performance output $y_p$ and the allowable value of $T_2$, we selected $T_2=0.1$, which leads, for each relaxation scheme, to $\gamma<1$.  
For such a value of $T_2$, in \figurename~\ref{fig:NonLinearForcedResponse}, we compare the components $\ep_3$ and $\ep_4$ of the solutions $\phi^a$, $\phi^b$, and $\phi^c$ to $\mathcal{H}_e$, obtained in correspondence to the gains designed via, respectively, Proposition~\ref{prop:P2:Chap3:propPred}, Proposition~\ref{prop:propX80}, and Proposition~\ref{prop:ZOH} from zero initial conditions in response to the following exogenous input $\tilde{w}\in\mathcal{L}_2$
${\footnotesize\tilde{w}(t)=\begin{cases}
	\sin(2t)&t\in[0,20]\\
	0
	\end{cases}}$. As in the former simulation, the $\tau$-component of all solutions coincide and satisfy  \eqref{eq:sequence_gen}.
The picture shows that the design based on Proposition~\ref{prop:P2:Chap3:propPred} provides the best result in terms of disturbance rejection.
\begin{figure}[!htbp]
	\centering
	\psfrag{t}[][][1.2]{$t$}
	\psfrag{ep3}[][][1.2]{$\ep_3$}
	\psfrag{ep4}[][][1.2]{$\ep_4$}
	\includegraphics[trim={1cm 0 0.2cm 0},clip, width=0.5\textwidth, clip]{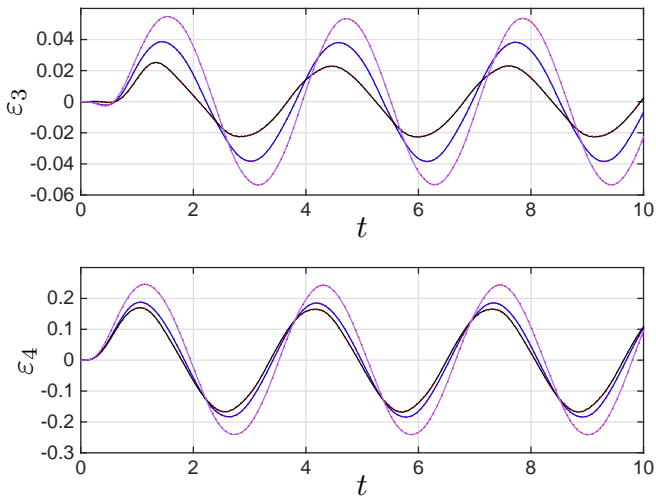}
	\caption{Evolution of $\phi^a$ (black), $\phi^b$ (blue), and $\phi^c$ (purple) projected onto ordinary time.}
	\label{fig:NonLinearForcedResponse}
\end{figure}

Before concluding this example, we want to show how our approach can be used to get an estimate of the largest allowable value of $T_2$ for a given design and how such an estimate compares with other approaches not relying on LMIs. More precisely, as pointed out in Remark~\ref{rem:pred}, by selecting $H=-CL$, the proposed observer 
coincides with the predictor-based observer scheme proposed in \cite{karafyllis2009continuous} and 
for which the results in \cite{postoyan2012framework} can be used to estimate the 
largest allowable value $T_2$ for a given gain $L$.
In particular, let us consider the following gain from \cite[Chapter 6.6.2]{postoyan2009Thesis} 
$$
L=\begin{pmatrix}
 9.328 & 1\\ -48.78 & 22.11\\ -0.0524 & 3.199\\ 19.41 & -0.9032
\end{pmatrix}
$$
and set $H=-CL$. 
An estimate of the largest allowable value $T_2$ for the given gains can be obtained by determining the largest value of $T_2$ for which \eqref{eq:M} are feasible. Notice that
when $L, H$, and $\delta$ are given, \eqref{eq:M} are LMIs, thus feasibility of those can be checked via semidefinite programming software.     
By picking $\lambda_t=0.01$, and by performing a line search on the scalar $\delta$, it turns out that \eqref{eq:M} are feasible for $T_2$ up to $0.1016$. In \cite[Chapter 6.6.2]{postoyan2009Thesis}, the authors show that the approach in \cite{postoyan2012framework} leads to an estimate of the largest allowable value of $T_2$ equal to $1.08\times 10^{-8}$. This shows how our approach allows one to get less conservative estimates of the largest allowable value of $T_2$. 
\end{example}
\section{Conclusion}
Building from the general ideas in \cite{karafyllis2009continuous}, this paper proposed a novel methodology to design, via linear matrix inequalities, an observer with intersample injection to exponentially estimate, with a given decay rate, the state of a continuous-time Lipschitz nonlinear system in the presence of sporadically available measurements. Moreover, the observer is robust to measurement noise, plant disturbances, and ensures a given level of performance in terms of  
$\mathcal{L}_2$-gain between plant exogenous disturbances and a given performance output.

Pursuing a unified approach, we provided several design methodologies to design the observer based on semidefinite programming. Two of them 
 lead back respectively to the observer scheme proposed in \cite{karafyllis2009continuous} and to the zero order sample-and-hold proposed in \cite{raff2008observer},  while the remaining lead to completely novel schemes. 
Several suboptimal design algorithms based on semidefinite programming are presented for the observer.
Numerical experiments underlined the significance of the proposed suboptimal design and showcased some interesting results of practical relevance. 
\appendix
\begin{lemma}
\label{lemma:KLj}
Let $(\phi,u)$ be a maximal solution pair to $\mathcal{H}_e$, and $\lambda_t$ be a strictly positive real number. Pick $(t, j)\in\dom\phi$ and let $0=t_0\leq t_1\leq t_2\leq\dots\leq t_{j+1}=t$ be such that $\dom\phi\cap\left([0,t]\times\{0,1,\dots,j\}\right)=\cup_{i=0}^j \left([t_i, t_{i+1}]\times \{i\}\right)$.  Then one has
$\sum_{i=1}^j e^{-2\lambda_t (t-t_i)}\leq \frac{e^{4\lambda_t T_1}}{e^{2\lambda_t T_1}-1}$.
\end{lemma}
\begin{proof}
First notice that for each $\mathbb{N}\ni i\leq j$, one has
\begin{equation}
\label{eq:SumTi}
t-t_i=\sum_{k=i+1}^{j} (t_{k}-t_{k-1})+(t-t_{j})
\end{equation} 
Pick $i\in\mathbb{N}$ with $i\leq j$, then from the structure of $\dom\phi$ given in \eqref{eq:Dom}, along with \eqref{eq:SumTi}, it follows that 
$t-t_i\geq \max\{0,(j-i-1)T_1\}$
which in turn yields
$$
\begin{array}{ll}
\sum_{i=1}^j e^{-2\lambda_t (t-t_i)}\leq& \sum_{i=1}^j e^{-2\lambda_t \max\{0,(j-i-1)T_1\}}\\
& \leq\sum_{i=1}^j e^{-2\lambda_t (j-i-1)T_1}\\
&=\frac{e^{4\lambda_t T_1}}{1-e^{2\lambda_t T_1}}  (e^{-2\lambda_t j T_1}-1)\\
&\leq\frac{e^{4\lambda_t T_1}}{e^{2\lambda_t T_1}-1}\\
\end{array}
$$
concluding the proof.
\end{proof}
\begin{lemma}
\label{claim:RangeM}
Let $P_1\in\Spnz,P_2\in\Spny$, $\delta, T_2,\chi$ be given positive scalars, and $L\in\mathbb{R}^{n_z\times n_y},H\in\mathbb{R}^{n_y\times n_y}$ be given matrices. For each $\tau\in[0, T_2]$ define $\mathcal{M}\colon\tau\mapsto \mathcal{M}(\tau)$. Then, $\range\mathcal{M}=\Co\{\mathcal{M}(0),\mathcal{M}(T_2)\}$.
\end{lemma}
\begin{proof}
To make the proof easier to follow, let us consider the following partitioning of the matrix $\mathcal{M}(\tau)$
$$
\scriptsize{
\mathcal{M}(\tau)=\left(\begin{array}{c|c|c}
\mathcal{M}_{1}&\mathcal{M}_{2}+e^{\delta\tau}\mathcal{M}_{3}&\mathcal{M}_{4}\\
\hline
\bullet&\mathcal{M}_{5}+e^{\delta\tau}\mathcal{M}_{6}&e^{\delta\tau}\mathcal{M}_{7}\\
\hline
\bullet&\bullet&\mathcal{M}_{8}
\end{array}\right)}
$$
where the corresponding blocks can be determined by simple comparison of the expression of $\mathcal{M}(\tau)$ given in \eqref{eq:M1}. 
Observe that for any $\tau\in[0,T_2]$, one has:
\begin{equation}
e^{\delta\tau}=\underbrace{\frac{e^{\delta\tau}-e^{\delta T_2}}{1-e^{\delta T_2}}}_{\lambda_1(\tau)}+\underbrace{\frac{1-e^{\delta\tau}}{1-e^{\delta T_2}}}_{\lambda_2(\tau)}e^{\delta T_2}
\label{eq:expDelta}
\end{equation}
where for each $\tau\in[0,T_1]$, $\lambda_1(\tau),\lambda_2(\tau)$ are nonnegative and such that $\lambda_1(\tau)+\lambda_2(\tau)=1$. Therefore, for each $\tau\in[0,T_2]$
\begin{equation}
\mathcal{M}(\tau)=\lambda_1(\tau)\mathcal{M}(0)+\lambda_2(\tau)\mathcal{M}(T_2)
\label{eq:M1Convex}
\end{equation}
which implies that $\range(\tau\mapsto\mathcal{M}(\tau))\subset\Co\{\mathcal{M}(0),\mathcal{M}(T_2)\}$. To conclude the proof, we show that $\range(\tau\mapsto\mathcal{M}(\tau))\supset\Co\{\mathcal{M}(0),\mathcal{M}(T_2)\}$. Pick
$\widetilde{M}\in\Co\{\mathcal{M}(0),\mathcal{M}(T_2)\}$, then there exists $\tilde{\lambda}\in[0,1]$ such that $\widetilde{M}=\tilde{\lambda}\mathcal{M}(0)+(1-\tilde{\lambda})\mathcal{M}(T_2)$. Pick
$$
\tilde{\tau}=\frac{\ln(\tilde{\lambda}(1-e^{\delta T_2})+e^{\delta T_2}))}{\delta}\in[0,T_2]
$$
and observe that from \eqref{eq:expDelta} one has $\lambda_1(\tilde{\tau})=\tilde{\lambda}$. Therefore, thanks to \eqref{eq:M1Convex}, one gets
$\mathcal{M}(\tilde{\tau})=\widetilde{M}$ and this concludes the proof.
\end{proof}
\balance
\bibliographystyle{plain}
\bibliography{biblio}

\begin{thebibliography}{10}

\bibitem{ahmed2012high}
T.~Ahmed-Ali and F.~Lamnabhi-Lagarrigue.
\newblock High gain observer design for some networked control systems.
\newblock {\em IEEE Transactions on Automatic Control}, 57(4):995--1000, 2012.

\bibitem{arcak2004framework}
M.~Arcak and D.~Ne{\v{s}}i{\'c}.
\newblock A framework for nonlinear sampled-data observer design via
  approximate discrete-time models and emulation.
\newblock {\em Automatica}, 40(11):1931--1938, 2004.

\bibitem{Boyd}
S.~Boyd, L.~E. Ghaoui, E.~Feron, and V.~Balakrishnan.
\newblock {\em Linear Matrix Inequalities in System and Control Theory}.
\newblock Society for Industrial and Applied Mathematics, June 1997.

\bibitem{boyd2004convex}
S.~Boyd and L.~Vandenberghe.
\newblock {\em Convex optimization}.
\newblock Cambridge university press, 2004.

\bibitem{cai2009characterizations}
C.~Cai and A.~R. Teel.
\newblock Characterizations of input-to-state stability for hybrid systems.
\newblock {\em Systems \& Control Letters}, 58(1):47--53, 2009.

\bibitem{ebihara2015}
Y.~Ebihara, D.~Peaucelle, and D.~Arzelier.
\newblock {\em S-variable Approach to LMI-based Robust Control}.
\newblock Springer, 2015.

\bibitem{farza2014continuous}
M.~Farza, M.~M'Saad, M.~L. Fall, E.~Pigeon, O.~Gehan, and K.~Busawon.
\newblock Continuous-discrete time observers for a class of mimo nonlinear
  systems.
\newblock {\em IEEE Transactions on Automatic Control}, 59(4):1060--1065, 2014.

\bibitem{FerranteIFAC2014}
F.~Ferrante, F.~Gouaisbaut, R.~G. Sanfelice, and S.~Tarbouriech.
\newblock An observer with measurement-triggered jumps for linear systems with
  known input.
\newblock In {\em Proceedings of the 19th World Congress of the International
  Federation of Automatic Control}, volume~19, pages 140--145, 2014.

\bibitem{ferrante2015hybrid}
F.~Ferrante, F.~Gouaisbaut, R.~G. Sanfelice, and S.~Tarbouriech.
\newblock A hybrid observer with a continuous intersample injection in the
  presence of sporadic measurements.
\newblock In {\em 2015 54th IEEE Conference on Decision and Control}, pages
  5654--5659, 2015.

\bibitem{Ferrante2016state}
F.~Ferrante, F.~Gouaisbaut, R.~G. Sanfelice, and S.~Tarbouriech.
\newblock State estimation of linear systems in the presence of sporadic
  measurements.
\newblock {\em Automatica}, 73:101--109, 2016.

\bibitem{fichera2012convex}
F.~Fichera, C.~Prieur, S.~Tarbouriech, and L.~Zaccarian.
\newblock A convex hybrid $\mathcal{H}_\infty$ synthesis with guaranteed
  convergence rate.
\newblock In {\em Proceedings of the 51st Annual Conference on Decision and
  Control (CDC)}, pages 4217--4222. IEEE, 2012.

\bibitem{FGZN_Auto2014}
F.~Forni, S.~Galeani, L.~Zaccarian, and D.~Ne{\v{s}}i{\'c}.
\newblock Event-triggered transmission for linear control over communication
  channels.
\newblock {\em automatica}, 50(2):490--498, 2014.

\bibitem{gahinet1994linear}
P.~Gahinet and P.~Apkarian.
\newblock A linear matrix inequality approach to {$H_\infty$} control.
\newblock {\em {International Journal of Robust and Nonlinear Control}},
  4(4):421--448, 1994.

\bibitem{goebel2009hybrid}
R.~Goebel, R.~G. Sanfelice, and A.~R. Teel.
\newblock Hybrid dynamical systems.
\newblock {\em Control Systems Magazine}, 29(2):28--93, 2009.

\bibitem{goebel2012hybrid}
R.~Goebel, R.~G. Sanfelice, and A.~R. Teel.
\newblock {\em Hybrid Dynamical Systems: Modeling, Stability, and Robustness}.
\newblock Princeton University Press, 2012.

\bibitem{hespanha2007survey}
J.~P. Hespanha, P.~Naghshtabrizi, and Y.~Xu.
\newblock A survey of recent results in networked control systems.
\newblock {\em Proceeding of the {IEEE}}, 95(1):138, 2007.

\bibitem{howell2002nonlinear}
A.~Howell and J.~K. Hedrick.
\newblock Nonlinear observer design via convex optimization.
\newblock In {\em Proceedings of the 2002 American Control Conference},
  volume~3, pages 2088--2093. IEEE, 2002.

\bibitem{0801.4824}
I.~Karafyllis and C.~Kravaris.
\newblock From continuous-time design to sampled-data design of nonlinear
  observers.
\newblock {\em arXiv:0801.4824}, 2008.

\bibitem{karafyllis2009continuous}
I.~Karafyllis and C.~Kravaris.
\newblock From continuous-time design to sampled-data design of observers.
\newblock {\em IEEE Transactions on Automatic Control,}, 54(9):2169--2174,
  2009.

\bibitem{lofberg2004yalmip}
J.~Lofberg.
\newblock Yalmip: A toolbox for modeling and optimization in matlab.
\newblock In {\em Computer Aided Control Systems Design, 2004 IEEE
  International Symposium on}, pages 284--289. IEEE, 2004.

\bibitem{Mayhew2013}
C.~G. Mayhew, R.~G. Sanfelice, and A.~R. Teel.
\newblock On path-lifting mechanisms and unwinding in quaternion-based attitude
  control.
\newblock {\em IEEE Transactions on Automatic Control}, 58(5):1179--1191, May
  2013.

\bibitem{mazenc2015design}
F.~Mazenc, V.~Andrieu, and M.~Malisoff.
\newblock Design of continuous--discrete observers for time-varying nonlinear
  systems.
\newblock {\em Automatica}, 57:135--144, 2015.

\bibitem{miettinen2012nonlinear}
K.~Miettinen.
\newblock {\em Nonlinear multiobjective optimization}, volume~12.
\newblock Springer Science \& Business Media, 2012.

\bibitem{moarref2014observer}
M.~Moarref and L.~Rodrigues.
\newblock Observer design for linear multi-rate sampled-data systems.
\newblock In {\em 2014 American Control Conference}, pages 5319--5324. IEEE,
  2014.

\bibitem{nevsic1999formulas}
D.~Ne{\v{s}}i{\'c}, A.~R. Teel, and E.~D. Sontag.
\newblock Formulas relating $\mathcal{KL}$ stability estimates of discrete-time
  and sampled-data nonlinear systems.
\newblock {\em Systems \& Control Letters}, 38(1):49--60, 1999.

\bibitem{nevsic2013finite}
D.~Ne{\v{s}}i{\'c}, A.~R. Teel, G.~Valmorbida, and L.~Zaccarian.
\newblock Finite-gain {$L_p$} stability for hybrid dynamical systems.
\newblock {\em Automatica}, 49(8):2384--2396, 2013.

\bibitem{Pipeleers:2009aa}
G.~Pipeleers, B.~Demeulenaere, J.~Swevers, and L.~Vandenberghe.
\newblock Extended {LMI} characterizations for stability and performance of
  linear systems.
\newblock {\em Systems \& Control Letters}, 58(7):510--518, 2009.

\bibitem{postoyan2009Thesis}
R.~Postoyan.
\newblock {\em Commande et construction d'observateurs pour les syst\`emes non
  lin\'eaires \`a donn\'ees \'echantillonn\'ees et en r\'eseau.}
\newblock PhD thesis, Univ. Paris-Sud, 2009.

\bibitem{postoyan2012framework}
R.~Postoyan and D.~Ne{\v{s}}i{\'c}.
\newblock A framework for the observer design for networked control systems.
\newblock {\em IEEE Transactions on Automatic Control}, 57(5):1309--1314, 2012.

\bibitem{postoyan2014tracking}
R.~Postoyan, N.~Van~de Wouw, D.~Ne{\v{s}}i{\'c}, and W.P. M.~H. Heemels.
\newblock Tracking control for nonlinear networked control systems.
\newblock {\em IEEE Transactions on Automatic Control}, 59(6):1539--1554, 2014.

\bibitem{raff2008observer}
T.~Raff, M.~Kogel, and F.~Allg\"{o}wer.
\newblock Observer with sample-and-hold updating for {L}ipschitz nonlinear
  systems with nonuniformly sampled measurements.
\newblock In {\em Proceedings of the American Control Conference}, pages
  5254--5257. IEEE, 2008.

\bibitem{sanfelice2013toolbox}
R.~G. Sanfelice, D.~Copp, and P.~Nanez.
\newblock A toolbox for simulation of hybrid systems in matlab/simulink: Hybrid
  equations (hyeq) toolbox.
\newblock In {\em Proceedings of the 16th international conference on Hybrid
  systems: computation and control}, pages 101--106. ACM, 2013.

\bibitem{spong1987modeling}
M.~W. Spong.
\newblock Modeling and control of elastic joint robots.
\newblock {\em Journal of dynamic systems, measurement, and control},
  109(4):310--319, 1987.

\bibitem{tutuncu2003solving}
R.~H. T{\"u}t{\"u}nc{\"u}, K.-C. Toh, and M.~J. Todd.
\newblock Solving semidefinite-quadratic-linear programs using {SDPT3}.
\newblock {\em Mathematical programming}, 95(2):189--217, 2003.

\bibitem{wang2017observer}
W.~Wang, D.~Ne{\v{s}}i{\'c}, and R.~Postoyan.
\newblock Observer design for networked control systems with flexray.
\newblock {\em Automatica}, 82:42--48, 2017.

\end{thebibliography}
\begin{IEEEbiography}[{\includegraphics[width=1in,keepaspectratio, clip]{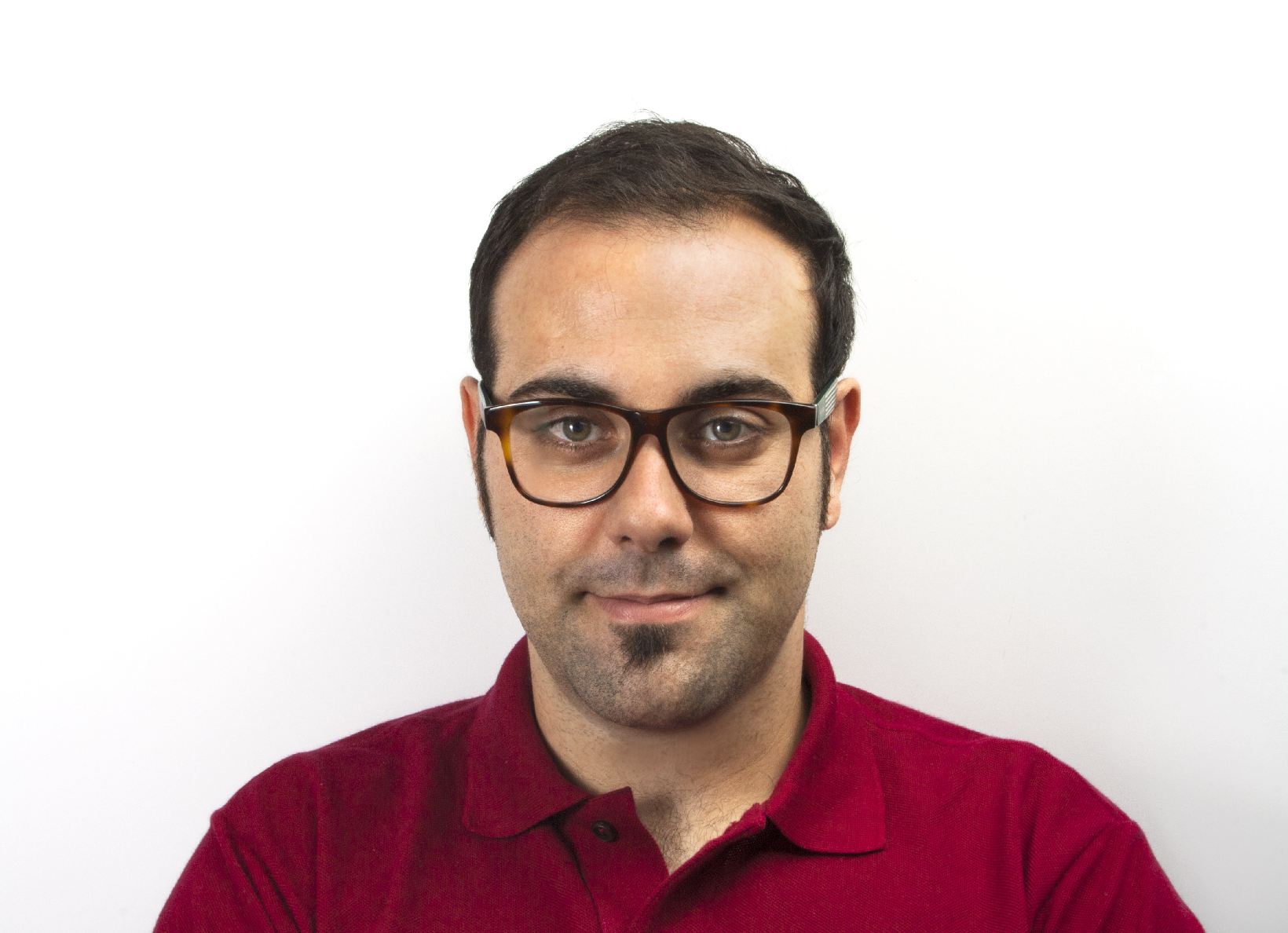}}]{Francesco Ferrante} received in 2010 a ``Laurea degre'' (BS) in Control Engineering from Sapienza University, Rome, Italy and in 2012 a ``Laurea Magistrale'' degree cum laude (MS) in Control Engineering from University Tor Vergata, Rome, Italy. 

During 2014, he held a visiting scholar position at the Department of Computer Engineering, University of California at Santa Cruz. In 2015, he received a PhD degree in Control Theory from Institut Sup\'erieur de l'A\'eronautique et de l'Espace (SUPAERO) Toulouse, France. From 2015 to 2017, he held postdoctoral positions at the Department of Electrical and Computer Engineering at Clemson University and at the Department of Computer Engineering at University of California Santa Cruz. In September 2017, he joined the University of Grenoble Alpes and the Grenoble Image Parole Signal Automatique Laboratory, where he is currently an assistant professor of Control Engineering. He is the recipient of the ``Best Ph.D. Thesis Award 2017'' from the Research Cluster on Modeling, Analysis and Management of Dynamical Systems (GDR-MACS) of the French National Council of Scientific Research (CNRS).
His research interests are in the general area of systems and control with a special focus on hybrid systems, observer design, and application of convex optimization in systems and control.\end{IEEEbiography}
\vspace{-1.5cm}

\begin{IEEEbiography}[{\includegraphics[width=1in,height=1.25in,clip,keepaspectratio]{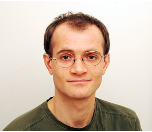}}]{Fr\'ed\'eric Gouaisbaut} was born in Rennes (France) in April 26, 1973.
He received the Dipl{\^o}me d'Ing\'enieur (Engineers' degree) from the \'Ecole Centrale de Lille, France, in September 1997 and the Dipl{\^o}me d'\'etudes Approfondies (Masters' Degree) from the University of Science and Technology of Lille, France, in September 1997. 
From October 1998 to October 2001 he was a Ph.D. student at the Laboratoire d'Automatique, G\'enie Informatique et Signal (LAGIS) in Lille, France. He received the Dipl{\^o}me de Doctorat (Ph.D. degree) from the \'Ecole Centrale de Lille and University of Science and Technology of Lille, France, in October 2001. Since October 2003, he is an associate professor at the Paul Sabatier University (Toulouse). His research interests include time delay systems, quantized systems, robust control and control with limited information.
\end{IEEEbiography}
\vspace{-1.2cm}

\begin{IEEEbiography}[{\includegraphics[width=1in,height=1.25in,clip,keepaspectratio]{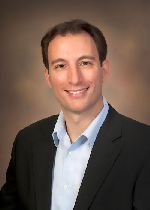}}]{Ricardo G. Sanfelice} is Professor at the Department of Electrical and Computer Engineering, University of California at Santa Cruz. He received the B.S. degree in Electronics Engineering from the Universidad Nacional de Mar del Plata, Buenos Aires, Argentina, in 2001. He joined the Center for Control, Dynamical Systems, and Computation at the University of California, Santa Barbara in 2002, where he received his M.S. and Ph.D. degrees in 2004 and 2007, respectively. During 2007 and 2008, he was a Postdoctoral Associate at the Laboratory for Information and Decision Systems at the Massachusetts Institute of Technology. He visited the Centre Automatique et Systemes at the Ecole de Mines de Paris for four months. From 2009 to 2014, he was Assistant Professor in the Aerospace and Mechanical Engineering at the University of Arizona, where he was also affiliated with the Department of Electrical and Computer Engineering and the Program in Applied Mathematics.	
\end{IEEEbiography}
\vspace{-1.5cm}

\begin{IEEEbiography}[{\includegraphics[width=1in,height=1.25in,clip,keepaspectratio]{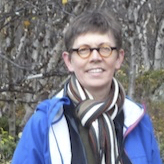}}]{Sophie Tarbouriech} received the PhD degree in Control Theory in 1991 and the HDR degree (Habilitation \`a Diriger des Recherches) in 1998 from University Paul Sabatier, Toulouse, France. 
Currently, she is full-time researcher (Directeur de Recherche) in LAAS-CNRS, Toulouse. Her main research interests include analysis and control of linear and nonlinear systems with constraints (limited information),
hybrid dynamical systems. She is currently Associate Editor for IEEE Transactions on Automatic Control, IEEE Transactions on Control Systems Technology, Automatica and European Journal of Control. 
She is also in the  Editorial Board of International Journal of Robust and Nonlinear Control. She is also co-Editor-in-Chief of the French journal JESA (Journal Europe\'en des Syst\`emes Automatis\'es). 
Since 1999, she is Associate Editor at the Conference Editorial Board  of the IEEE Control Systems Society. She is also a member of the IFAC technical committee on Robust Control and Nonlinear Systems. 	
\end{IEEEbiography}
\end{document}